\documentclass[11pt,reqno]{article}   
\usepackage{fullpage}

\usepackage[unicode=true]{hyperref}
\usepackage{amsmath}
\usepackage{cite}
\usepackage{amsfonts}
\usepackage{amssymb}
\usepackage{amsthm}
\usepackage{authblk}
\usepackage{mypack} 
\usepackage[pdftex]{color,graphicx}
\usepackage{adjustbox}
\usepackage{soul}
\usepackage{comment}
\usepackage{bbold}
\usepackage{tikz} 
\usetikzlibrary{decorations.pathreplacing,angles,quotes}
\usetikzlibrary{matrix,arrows,decorations.pathmorphing}
\usepackage{tikz-cd}
\usetikzlibrary{arrows}
\usetikzlibrary{decorations.markings}

\usepackage{bbold}
\usepackage{diagbox}
\usepackage{caption}
\usepackage{subcaption}
\usepackage{pdfpages} 
\usepackage{blkarray}
\usepackage{centernot}
\usepackage{mathtools}

\definecolor{bluegray}{rgb}{0.4, 0.6, 0.8}
\definecolor{turquoise}{rgb}{0.2, 0.7, 0.6}

\def\on{\operatorname}
\def\sf{\mathsf}

\title{Twisted simplicial distributions}

\author{Cihan Okay\footnote{cihan.okay@bilkent.edu.tr}} 
\author{Walker H. Stern\footnote{walker@walkerstern.com }}
\affil{{\small{Department of Mathematics, Bilkent University, Ankara, Turkey}}}

\date{\today}

\begin{document}
  \maketitle


\begin{abstract}
	We introduce a theory of twisted simplicial distributions on simplicial principal bundles, which allow us to capture Bell's non-locality, 
{and}	
	the more general notion of quantum contextuality. We leverage the classical theory of simplicial principal bundles, as well as structures on categories of such bundles, to provide powerful computational tools for analyzing twisted distributions in terms of both direct constructions in simplicial sets and techniques from homological algebra. We use these techniques to analyze our key examples: quantum distributions and operator-theoretic polytopes used in the classical simulation of quantum computation.   
\end{abstract}

  \tableofcontents

\section{Introduction}
 Bell's notion of non-locality and its generalization quantum contextuality are fundamental features of quantum theory \cite{bell64,KS67}.  
In \cite{okay2022simplicial}, the first author introduced a new, simplicial approach for contextuality based on simplicial distributions. 
This approach subsumes two earlier prominent approaches: the sheaf-theoretic approach of Abramsky--Brandenburger \cite{abramsky2011sheaf} and the cohomological approach of Okay et al. \cite{Coho}.   
The theory of simplicial distributions combines the sheaf-theoretic and cohomolog{ical approaches by modeling measurements and their outcomes as simplicial sets. 
One novel feature of the theory is that contextuality proofs can be done using arguments that rely on topological intuition, such as gluing and extending distributions \cite{kharoof2023topological}. In this framework, contextuality can be characterized using homotopy \cite{HomVert}, and categorical methods can be easily applied  \cite{barbosa2023bundle}.
   
In this work, we introduce simplicial distributions that are 
{``twisted"}
by local algebraic constraints --- a notion that has immediate applications to studying naturally appearing twisted quantum probabilities. Our framework is based on the theory of principal bundles from simplicial homotopy theory. The simplicial theory is analogous to the topological theory, \emph{but is in some ways more computationally tractable}. 
It is also of interest to note that the general simplicial principal bundles which appear in our work are also of interest elsewhere in mathematical physics, as they provide explicit models of principal $\infty$-bundles \cite{NSS}. Some of the many applications of principal $\infty$-bundles to mathematical physics, particularly string theory, are surveyed in \cite{bunk2023inftybundles}.
    
The theory of simplicial distributions introduced in \cite{okay2022simplicial} models contextuality by representing the measurements and outcomes of an experimental setting by simplicial sets.  
A simplicial distribution is a simplicial set map
$$
p:X\to D(Y).
$$  
Here, $D$ denotes the distribution monad \cite{jacobs2010convexity}, which also lifts to the category of simplicial sets \cite{kharoof2022simplicial}. Explicitly, an $n$-simplex of $D(Y)$ is simply a probability distribution on $Y_n$. 
In effect, a simplicial distribution is a collection of distributions $\set{p_x\in D(Y_n):x\in X_n}_{n\geq 0}$ compatible under the simplicial structure maps.
This compatibility condition can be seen as a topological version of the physical {\it non-signaling conditions}.   

Twisted distributions are defined on principle $K$-bundles for a simplicial group $K$. 
We will model principal bundles using twisted products. 
A twisting $\eta$ consists of a collection of functions $\set{\eta_n:X_n\to K_{n-1}}_{n\geq 0}$.
  The twisted product $K\times_\eta X$ has the same set of simplices as the Cartesian product $K\times X$, and the simplicial structure maps are the same, except that the $d_0$-face is twisted by  $\eta$. 
A typical case of interest is when $K$ is the nerve space $NH$ of an Abelian group, in which case 
the twisting is completely determined by a $2$-cocycle $\gamma:X_2 \to H$. In dimension $2$ the twisting at the $d_0$-face is given by 
$$
d_{{0}}(h_1,h_2;x) =  (h_2+\gamma(x),d_0(x)). 
$$  
Projecting a twisted product gives a principal bundle, denoted by $\pi:K\times_\eta X\to X$.
Following
the interpretation in \cite{barbosa2023bundle}, the base space $X$ represents the space of measurements as before {and t}he fiber $x^*(E)$ over a simplex $x:\Delta^n\to X$ represents the space of outcomes for the measurement $x$.
A twisted simplicial distribution on $\pi$ is a simplicial set map $p:X\to D(E)$ making the following diagram commute:
  $$
  \begin{tikzcd}
  	& D(E) \arrow[d,"{D(\pi)}"] \\
  	X \arrow[ru,"p"] \arrow[r,"\delta"'] & D(X)
  \end{tikzcd}
  $$  
where $\delta:X\to D(X)$ sends a simplex to the delta-distribution on that simplex. More explicitly, a twisted simplicial distribution $p$ consists of a family of distributions $\set{p_x\in D(K_n):\, x\in X_n}_{n\geq 0}$ compatible with the face maps. In dimension $2$, the $d_0$-face will be twisted as described above.  
To work out the effect of the twisting 
 for higher dimensions, we provide an explicit formula in Lemma \ref{lem:twist_fxn_norm_cocycle}. 
  Writing $\sDist_\eta(X)$ for the set of $\eta$-twisted distributions, we show in Corollary \ref{cor:twiste dist is a polytope} that this set has the structure of a convex polytope with finitely many vertices under a mild finiteness condition on the measurement space. 
Twisted distributions for higher dimensions and the corresponding vertex enumeration problem are unexplored territory, including quantum mechanical distributions and polytopes used in the classical simulation of quantum computation. Key examples include, for instance, the $\Lambda$ polytopes of \cite{zurel2020hidden} and the toy example \cite{okay2022mermin} studied in Section \ref{sec:mermin polytope}.

Principal bundles over $X$ with fiber $K$ can be assembled into a groupoid $\sf{Bun}_K(X)$ whose morphisms are given by isomorphisms of bundles.  
Moreover, given two bundles $E$ and $F$, there is a tensor product $E\otimes_K F$, which equips the category with the structure of a symmetric monoidal groupoid (Theorem \ref{thm:bun is SMC}).
Moreover, this groupoid is a commutative $2$-group, which means that every principal bundle has an inverse up to isomorphism under the tensor product. For twisted products the inverse of $K\times_\eta X$ is given by $K\times_{\eta^{-1}}X$.
Using the tensor product, we can define a convolution product of twisted distributions that takes two distributions $p$ and $q$ on the bundles and produces a distribution $p\ast q$ on their tensor product.  
 The functor
  $\sDist: \sf{Bun}_K(X) \to \catSet$  
  {which sends a principal $K$-bundle to its set of simplicial distributions} is lax symmetric monoidal with respect to the tensor product of bundles and Cartesian product of sets (Proposition \ref{pro:sDist is lax SMF}). Then, a monoidal version of the Grothendieck construction \cite{moellervasilakopoulou,Omongroth} produces a symmetric monoidal category
  $$
  \int_{\sf{Bun}_K(X)} \sDist
  $$ 
  whose objects are twisted simplicial distributions on bundles with the monoidal structure given by the convolution product (Corollary \ref{cor:Grothendieck construction}). In the rest of the paper we heavily rely on the categorical structures identified on bundles and distributions defined on them. 
 
Our first key result is the identification of twisted distributions as equivariant simplicial distributions (Theorem \ref{thm:twisted as equivariant}).
We consider a measurement space $X$ and an outcome space $Y$ both equipped with a $K$-action. 
Then $D(Y)$ also inherits a $K$-action given by the convolution product. 
The set $\sDist_K(X,Y)$ of $K$-equivariant simplicial distributions consists of simplicial set maps $X\to D(Y)$ that are $K$-equivariant. 

\begin{thm*}
There is a convex bijection
	\[
	\sDist_\eta(X)\cong \sDist_K(K\times_{\eta^{-1}} X, K).
	\] 
\end{thm*}

The main examples of twisted distributions come from quantum distributions. There is a quantum version of the probability monad where the semiring $\RR_{\geq 0}$ is replaced by projectors $\Proj(\hH)$ on a finite-dimensional Hilbert space $\hH$. A projective measurement on a set $X$ is given by a function $\Pi:X\to \Proj(\hH)$ with finite support that sums to the identity operator.
The Born rule of quantum theory provides a way to extract probabilities from a projective measurement. In effect, this allows us to represent quantum states as simplicial distributions $P_\hH(N\ZZ_d)\to D(N\ZZ_d)$ that are equivariant with respect to a natural $N\ZZ_d$ action.
Then, using the theorem above, we can regard these equivariant distributions as twisted simplicial distributions on a suitable quotient of $P_\hH(N\ZZ_d)$.

We introduce the notion of contextuality for twisted distributions in Section \ref{sec:Contextuality} following the bundle perspective of \cite{barbosa2023bundle}.
A deterministic simplicial {distribution} is of the form of a delta distribution peaked at a section of the bundle. Non-contextual distributions are those that can be expressed as a probabilistic mixture of deterministic distributions. Those that cannot are called contextual.
Therefore, this definition of contextuality is useful only in the untwisted case, that is when the principal bundle is trivial. To remedy this, we introduce a relative version of contextuality associated to a trivializing morphism $U\to X$, i.e., a morphism such that the pullback of the bundle along it is trivial.
We discuss natural examples of such morphisms. 
Twisting can be used as an effective tool to simplify the measurement space for studying contextuality. More precisely, we have a ``twist and collapse" result (Theorem \ref{thm:twist collapse iso for twistings}) with respect to a simplicial subset $i:Z\to X$. 
Let $\varphi_0$ denote the zero section of the projection $K\times Z\to Z$. Consider the set $\sDist_\eta(X,0)$ of simplicial distributions that restrict to the delta distribution $\delta^{\varphi_0}$ on $Z$.

\begin{thm*} 
	There is a convex isomorphism 
	\[
	\sDist_\eta(X/Z)\cong \sDist_\eta(X,0). 
	\]
\end{thm*}

 This result has important computational applications, mainly for the computation of the vertices of the polytope of twisted distributions. 
 In the case of $2$-dimensional measurement spaces, such computations in the context of non-signaling polytopes and their twisted versions{, e.g., the Mermin polytope,} are provided in \cite{okay2023rank}. 

The rest of the paper is organized as follows. In Section \ref{sec:Simplicial principal bundles}, we review the classical theory of principal bundles of simplicial sets. 
We describe twisted products in the case of an Abelian simplicial group in Section \ref{sec:The theory for a simplicial Abelian group} and then further specialize to the case where the simplicial group is the nerve of an Abelian group in Section \ref{sec:The case of an Abelian group}. Then, we turn to tensor products of principal bundles and prove our main category-theoretic results. Twisted distributions, the main objects of interest in this paper, are introduced and studied in Section \ref{sec:Twisted distributions}. 
The identification of twisted distributions with equivariant distributions is proved in this section. As {an} application, quantum distributions (Example \ref{ex:quantum}) are realized as twisted distributions. The notion of contextuality and its new 
{relative}
variant are introduced in Section \ref{sec:Contextuality}. Our twisting and collapse theorem is proved in this section. As a concrete example, we provide an in-depth analysis of the Mermin polytope in Section \ref{sec:mermin polytope}.

\paragraph{Acknowledgments.}
This work is supported by the Air Force Office of Scientific Research under
award number FA9550-21-1-0002.
The first author also acknowledges support from the Digital Horizon Europe project FoQaCiA, GA no. 101070558.

\section{Simplicial principal bundles} 
\label{sec:Simplicial principal bundles}

In this section, we will briefly recapitulate the theory of simplicial principal bundles. In most cases the theory is wholly analogous to that of principal bundles over a topological group, however, the details of implementation often differ. We will attempt to make clear the parallels in the general theory while presenting these details here. One important way in which the presentation here differs from the original topological setting is that we only consider \emph{left} actions, rather than the right actions more common in the setting of topological principal bundles. In the Abelian setting where much of our work will take place, this distinction evaporates, and so the reader familiar with, e.g. \cite[\S 18]{May67} may safely read those later sections without reference to these. 

\subsection{Preliminaries}

Before embarking on our study of principal bundles, we briefly recall some of the key definitions and notations from simplicial homotopy theory.

\Def{\label{def:simplex-category} 
	The \emph{simplex category} $\Delta$ consists of
	\begin{itemize}
		\item objects $[n]=\set{0,1,\cdots,n}$ where $n\geq 0$, and
		\item morphism $\theta:[n]\to [m]$ given by order preserving functions, i.e., $\theta(i)\leq \theta(j)$ whenever $i\leq j$. 
	\end{itemize}
}

Given a category $\catC$ a \emph{simplicial object} in this category is a functor
$$
\begin{tikzcd}
	X:&[-3em]\Delta^\op \arrow[r] & \catC.
\end{tikzcd}
$$
The category $\catsC$ of simplicial objects in $\catC$ has morphisms given by natural transformations.
We will be primarily concerned with simplicial sets, simplicial groups, and simplicial Abelian groups, that is, the cases where $\catC=\catSet$, $\catC=\catGrp$, and $\catC=\catAb$, respectively. When working with simplicial \emph{sets}, we will use the notation $\Delta^n$ for the image of $[n]$ under the Yoneda embedding, i.e., the functor $\Delta(-,[n]):\Delta^\op\to \catSet$. The simplicial sets $\Delta^n$ are called the \emph{standard $n$-simplices}. As a presheaf category, the category of simplicial sets is complete and cocomplete, with limits and colimits given pointwise. 

\begin{ex}
	Given a group $G$, we will denote by $NG$ the \emph{nerve}\footnote{Using the usual categorical definition of the nerve, this is the nerve of the one-object groupoid $BG$.} of the group $G$, which is the simplicial set with sets of $n$-simplices given by 
	\[
	NG_n= G^{n}.
	\]
	The face maps are given by
	$$
	d_i(a_1,\cdots,a_n) = \left\lbrace
	\begin{array}{ll}
		(a_2,\cdots,a_n) & i=0\\
		(a_1,\cdots,a_i\cdot a_{i+1},\cdots, a_n) & 0<i<n \\
		(a_1,\cdots,a_{n-1}) & i=n
	\end{array}
	\right.
	$$
	and the degeneracy maps are given by
	$$
	s_j(a_1,\cdots,a_n) = (a_1,\cdots,a_{{j}},1,a_{{j+1}},\cdots, a_n).
	$$
	Note that in the special case where $G$ is Abelian, the nerve $NG$ is a simplicial Abelian group. 
\end{ex}

\begin{defn}
	Let $X$ and $Y$ be simplicial sets, and $f,g:X\to Y$ be maps of simplicial sets. A \emph{simplicial homotopy from $f$ to $g$} from $f$ to $g$ is a morphism of simplicial sets 
	\[
	\begin{tikzcd}
		H: &[-3em] X\times \Delta^1 \arrow[r] & Y 
	\end{tikzcd}
	\]
	such that restricting $H$ to the simplicial subset $X\cong X\times \{0\}$ yields $f$, and restricting $H$ to the simplicial subset $X\cong X\times\{1\}$ yields $g$. If there is a simplicial homotopy from $f$ to $g$, we write $f\sim g$ and say that $f$ and $g$ are \emph{(simplicially) homotopic}. We denote the set of equivalence classes under simplicial homotopy by $[X,Y]$. 
\end{defn}

\begin{rem}
	One must be cautious when working with simplicial homotopy, as it is not, in general, an equivalence relation on $\catsSet(X,Y)$. In the case where $Y$ is a so-called \emph{Kan complex} --- for instance, the nerve of a group --- this difficulty evaporates. We will not define or discuss Kan complexes in this work, however, some of the results from simplicial homotopy theory which we cite will rest upon the fact that many of our constructions, most notably $\overline{W}K$, are Kan complexes. 
\end{rem}


\subsection{Simplicial principal bundles}

Throughout this section, let $K$ denote a simplicial group, which may be equivalently interpreted as a group object in $\catsSet$, or as a simplicial object in $\catGrp$. 

\begin{defn}
	Let $K$ be a simplicial group. A \emph{$K$-principal bundle} is a simplicial set $E$ equipped with a left\footnote{This is, as remarked above, the major difference between this treatment and \cite{May67}. We choose instead to follow the conventions in \cite{goerss2009simplicial}.} $K$-action $m$ which is \emph{free}, i.e., the diagram 
	\[
	\begin{tikzcd}
		K\times E \arrow[r,"m"]\arrow[d,"\operatorname{pr}_2"'] & E\arrow[d,"\pi"] \\
		E\arrow[r,"\pi"'] & X 
	\end{tikzcd}
	\]
	is pullback, where $\pi:E\to X$ is the quotient map. Note that this is equivalent to each $E_k$ being a free $K_k$-set for every $k\geq 0$. 
\end{defn}

\begin{defn}
	A \emph{morphism of principal $K$-bundles over $X$} from $\pi_E:E\to X$ to $\pi_F:F\to X$ is a $K$-equivariant map of simplicial sets $f:E\to F$ such that the diagram 
	\[
	\begin{tikzcd}
		E\arrow[dr,"\pi_E"']\arrow[rr,"f"] & & F\arrow[dl,"\pi_F"] \\
		 & X &
	\end{tikzcd}
	\]
	commutes. It is immediate from the freeness of the action that every such morphism is an isomorphism. We will denote by $\sf{Bun}_K(X)$ the groupoid of principal $K$-bundles over $X$. We will write $\on{Prin}_K(X)=\pi_0(\sf{Bun}_K(X))$ for the set of isomorphism classes of principal $K$-bundles over $X$.
\end{defn}

As in the topological case, there are well-defined notions of classifying spaces and universal bundles associated to simplicial groups. We will make heavy use of one model for the universal bundle, denoted by $WK\to \overline{W}K$, which we now recall. 

\begin{defn}
	Let $WK$ be the simplicial set with 
	\[
	WK_n:= K_n\times \cdots \times K_0 
	\]
	with structure maps 
	\[
	d_i(g_n,\ldots,g_0)= \begin{cases}
		(d_ig_n,d_{i-1}g_{n-1},\ldots, (d_0g_{n-i})g_{n-i-1},g_{n-i-2},\ldots, g_0) & i<n\\
		(d_ng_n,\ldots,d_1g_1) & i=n
	\end{cases}
	\]
	and 
	\[
	s_i(g_n,\ldots,g_0)=(s_ig_n,\ldots,s_0g_{n-i},1,g_{n-i-1},\ldots,g_0)
	\]
	We denote by $\overline{W}K$ the quotient of $WK$ by the obvious left $K$-action. More explicitly, $\overline{W}K$ has $n$-simplices given by 
	\[
	\overline{W}K_n:= K_{n-1}\times \cdots \times K_0. 
	\]
	and structure maps 
	\[
	d_i(g_{n-1},\ldots,g_0) = \begin{cases}
		(g_{n-2},\ldots,g_0) & i=0\\
		(d_{i-1}g_{n-1},\ldots, (d_0g_{n-i})g_{n-i-1},g_{n-i-2},\ldots, g_0) & 0<i<n\\
		(d_{n-1}g_{n-1},\ldots,d_1g_1) & i=n
	\end{cases}
	\]
	and 
	\[
	s_i(g_{n-1},\ldots,g_0)=\begin{cases}
		(1,g_{n-1},\ldots, g_0) & i=0\\
		(s_{i-1}g_{n-1},\ldots s_0g_{n-i},1,g_{n-i-1},\ldots, g_0) & i>0. 
	\end{cases}
	\]
	We denote by $\pi:WK\to \overline{W}K$ the quotient map. 
\end{defn} 

\begin{thm}\label{thm:characterization_universal_bundle}
	For any simplicial set $X$, pullback induces a bijection 
	\[
	\begin{tikzcd}
		{[X,\overline{W}K]}\arrow[r,"\cong"] & \Prin_K(X).
	\end{tikzcd}
	\]
\end{thm}

\begin{proof}
	This theorem is established for any principal $K$-bundle $EK\to BK$ satisfying certain conditions in \cite[Ch. V, Thm. 3.9]{goerss2009simplicial}. It is shown that $\pi:WK\to \overline{W}K$ satisfies these conditions in \cite[Ch. V, Cor. 6.8]{goerss2009simplicial} and \cite[Ch. V, Lem. 4.6]{goerss2009simplicial}.
\end{proof}

By Theorem \ref{thm:characterization_universal_bundle}, there is a bijection 
\[
\operatorname{Prin}_K(X)\cong [X,\overline{W}K]
\]
for any $X\in \catsSet$. However, since the action of $K$ on $E$ is free for any principal bundle $E\to X$, we should be able to identify simplices of $E$ with pairs $(g,x)$ where $g\in K_n$ and $X$ in $X_n$. Fortunately, the fiber of $WK\to \overline{W}K$  over $(g_{n-1},\ldots,g_0)$ can be canonically identified with $K_n$ by construction. We now seek to untangle the specific structure maps in the pull-back
\[
 X_n \cong (WK\times_{\overline{W}K} X)_n
\]
for a given map $\gamma:X\to \overline{W}K$. 

A quick computation of the structure maps shows that, for $(g,x)\in  K_n\times X_n$ 
\[
d_i(g,x)=\begin{cases}
	(d_i(g),d_i(x)) & i>0\\
	(d_0(g)\cdot\gamma_n(x),d_0(x)) & i=0
\end{cases}
\]
and 
\[
s_i(g,x)=(s_i(g),s_i(x)).
\]
However, this implies something curious --- the bundle $WK\times_{\overline{W}K} X$ only depends on the   first component of $\gamma$, and so the homotopy class of $\gamma$ only depends on this first component.  

To understand how this works, we turn it into a definition. 

\begin{defn}
	For $X\in \catsSet$, a collection $\eta:=\{\eta_n:X_n\to K_{n-1}\}_{n\geq 1}$ is called a \emph{twisting function} if the maps 
	\[
	d_i(g,x)=\begin{cases}
		(d_i(g),d_i(x)) & i>0\\
		(d_0(g)\cdot\eta_n(x),d_0(x)) & i=0
	\end{cases}
	\]
	and 
	\[
	s_i(g,x)=(s_i(g),s_i(x))
	\]
	define a simplicial structure on the graded set $ K_n\times X_n$. Direct computation of the simplicial identities then shows that $\eta$ is a twisting function if and only if the identities
	\[
	d_i\eta(x)=\begin{cases}
		\eta(d_{i+1}(x)) & i>0\\
		\eta(d_1(x))\cdot \eta(d_0(x))^{-1} & i=0 
	\end{cases}
	\] 
	and 
	\[
	\eta(s_j(x))=\begin{cases}
		1 & j=0\\
		s_{j-1}(\eta(x)) & \text{else}.
	\end{cases}
	\]
	Given a twisting function $\eta$, we call the associated simplicial set the \emph{$\eta$-twisted product of $X$ and $K$}, and denote it by $K\times_\eta X$.  We denote the set of twisting functions on $X$ by $\on{Twist}_K(X)$. We denote the projection map by $\pi_\eta:K\times_{\eta}X\to X$. 
\end{defn}

\begin{lem} 
	Given a twisting function $\eta$ on $X$, the map $K\times_\eta X\to X$ is a principal $K$-bundle. 
\end{lem}

\begin{proof}
	It is immediate that the left action of $K$ on $K\times_\eta X$ is free, and that the projection induces an isomorphism from the quotient. 
\end{proof}

Our final step is to determine when two twisting functions determine homotopic maps to $\overline{W}K$. 

\begin{lem}\label{lem:classifying_map_of_twisting_function}
	Given a twisting function $\eta$, the map 
	\[
	\begin{tikzcd}
		\theta^\eta: &[-3em] K\times_{\eta} X \arrow[r] & WK
	\end{tikzcd}
	\]
	given by 
	\[
	\theta_n^\eta(g,x)=(g,\eta_n(x),\eta_{n-1}(d_0(x)),\ldots, \eta_1(\underbrace{d_0\cdots d_0}_{\times n-1}(x)))
	\]
	is a $K$-equivariant simplicial map and the diagram 
	\[
	\begin{tikzcd}
		K\times_\eta X \arrow[d]\arrow[r,"\theta"] & WK\arrow[d]\\
		X \arrow[r,"\overline{\theta}"'] & \overline{W}K 
	\end{tikzcd}
	\]
	is pullback, where $\overline{\theta}$ denotes the induced map on quotients. 
\end{lem}

\begin{proof}
	This follows by direct computation using the simplicial identities. A proof can be found in Appendix \ref{sec:simp_comp}.
\end{proof}

We then aim to classify when two twisted products are equivalent. This is simplified by the fact that, by freeness, two bundles over the same base are isomorphic if and only if there is a bundle map between them. Thus, if there is an isomorphism 
\[
\begin{tikzcd}
	f: &[-3em] K\times_{\eta} X\arrow[r] & K\times_\tau X 
\end{tikzcd}
\]
over $X$, it will have the form 
\[
f(g,x)=(x,\kappa(g,x))
\]
and $\kappa_n(g,x)$ will be $K_n$-equivariant, i.e, 
\[
\kappa_n(g,x)= g\cdot\psi_n(x)
\]
for some function $\psi$.

\begin{lem}\label{lem:twisting_equiv}
	Given functions $\psi_n: X_n\to K_n$, and the maps $f_n:(K\times_\eta X)_n\to (K\times_\tau X)_n$ defined by the $\psi_n$, the following are equivalent.
	\begin{enumerate}
		\item The map $f$ is a map of simplicial sets.
		\item The map $f$ is an isomorphism of principal $K$-bundles over $X$. 
		\item The maps $\psi_i$ satisfy the relations
		\[
		\begin{aligned}
			d_0(\psi_n(x))\eta_n(x)&= \eta_n(x)\psi_{n-1}(d_0(x))\\
			\psi_{n-1}(d_i(x)) & =d_i(\psi_n(x)) &  i>0\\
			s_i\psi_n(x) & =\psi_{n-1}(s_i(x)) &i\geq 0.
		\end{aligned}
		\]
	\end{enumerate}
\end{lem}

\begin{proof}
	The equivalence of (1) and (2) is immediate, since equivariant maps between free $G$-sets are isomorphisms. We thus need only show the equivalence with (3). 
	
	The map $f$ being simplicial is equivalent to the relations 
	\[
	\begin{aligned}
		d_i f(g,x)& =f(d_i(g,x))\\
		s_i f(g,x) & = f(s_i(g,x)).
	\end{aligned}
	\]
	Expanding these out in terms of the definition of $\psi_n$ and comparing group terms yields the relations of the lemma. 
\end{proof}

\begin{defn}
	We call twisting functions $\eta$ and $\tau$ \emph{equivalent} if there is a collection of functions $\psi_n$ satisfying the conditions of Lemma \ref{lem:twisting_equiv}.
\end{defn}

\begin{cor}\label{cor:twistings_are_PBs}
	The map 
	\[
	\begin{tikzcd}[row sep=0em]
		\theta^{(-)}:&[-3em] \on{Twist}_K(X)_{/\sim}\arrow[r] & {[X,\overline{W}K]}\\
		& \eta \arrow[r,mapsto] & \overline{\theta^{\eta}}
	\end{tikzcd}
	\]
	(see Lemma \ref{lem:classifying_map_of_twisting_function}) is a bijection, where $\sim$ denotes equivalence of twisting functions.
\end{cor}

\begin{lem}\label{lem:TwisPro}
	The levelwise sum of two twisting functions $\eta$ and $\eta^\prime$ is a twisting function.
\end{lem}

\begin{proof} 
	Given $x \in X_n$, we have
	$$
	\begin{aligned}
		d_0((\eta +\eta')_n(x))
		&=d_0(\eta_n(x)+\eta'_n(x))\\
		&=d_0(\eta_n(x))+d_0(\eta'_n(x))\\
		&=
		\eta_{n-1}(d_1(x))-\eta_{n-1}(d_0(x))+
		\eta'_{n-1}(d_1(x))-\eta'_{n-1}(d_0(x))
		\\
		&=(\eta+\eta')_{n-1}(d_1(x))-(\eta+\eta')_{n-1}(d_0(x)).
	\end{aligned}
	$$
	For $1 \leq i \leq n$, we have 
	$$
	\begin{aligned}
		d_i((\eta +\eta')_n(x))
		&=d_i(\eta_n(x)+\eta'_n(x))\\
		&=d_i(\eta_n(x))+d_i(\eta'_n(x))
		\\
		&=
		\eta_{n-1}(d_{i+1}(x))+
		\eta'_{n-1}(d_{i+1}(x))
		\\
		&=(\eta+\eta')_{n-1}(d_{i+1}(x)).
	\end{aligned}
	$$
	For the degeneracy maps we have
	$$
	(\eta+\eta')_n(s_0(x))
	=\eta_n(s_0(x))+\eta'_n(s_0(x))=0,
	$$
	and for $1 \leq j \leq n+1$, 
	$$
	\begin{aligned}
		(\eta + \eta')_n(s_j(x))
		&=\eta_n(s_j(x))+\eta'_n(s_j(x))\\
		&=
		s_{j-1}(\eta_{n-1}(x))+s_{j-1}(\eta'_{n-1}(x))\\
		&=s_{j-1}((\eta+\eta')_{n-1}(x)).
	\end{aligned}
	$$
\end{proof}

\Lem{
	If $K$ is a simplicial Abelian  group, then $\Twist_K(X)$ is an Abelian group. 
}
\Proof{
	By Lemma \ref{lem:TwisPro}
	the sum is well-defined. The associativity follows from the associativity in $K_n$ for every $n\geq 0$. The identity element is the twisting function that send every $x \in X_n$ to the identity elment of $K_{n-1}$. The inverse of $\eta$ is given by $-\eta$.
}

\begin{defn}\label{def:ProSec}
	We will write $\Sec(f)$ for the set of sections of a simplicial set map $f:E\to X$. Note that when $E$ is given by a twisted product $E=K\times_\eta X$, we have a canonical identification of a section $s:X\to E$ with a map of graded sets $\varphi: X\to K$ satisfying 
	\[
	\begin{aligned}
		d_i\varphi(x)&=\varphi(d_i(x))  &  i>0\\
		s_i\varphi(x)& =\varphi(s_i(x))  & \\
		d_0(\varphi(x))\eta(x) &= \varphi(d_0 x). & 
	\end{aligned} 
	\]
	Under this identification, $s(x)=(\varphi(x),x)$. We will typically denote an element of $\Sec(\pi_\eta)$ by the function $\varphi:X\to K$, rather than the section $s$ itself. When necessary, we will denote the associated section by $s_\varphi$.

	We define a product 
	\[
	\begin{tikzcd}
		\kappa:&[-3em] \Sec(\pi_\eta) \times \Sec(\pi_{\eta'}) \arrow[r] & \Sec(\pi_{\eta+\eta'})
	\end{tikzcd}
	\]
	by sending $(\varphi,\varphi^\prime)$ to the pointwise sum $\varphi+\varphi^\prime$.
\end{defn}

\begin{pro}\label{prop:triv_bundle=has_section}
	Let $\eta$ be a twisting function on $X$. The following are equivalent. 
	\begin{enumerate}
		\item The set $\Sec(\pi_\eta)$ is non-empty.
		\item There is an isomorphism $K\times_\eta X\cong K\times X$ of $K$-bundles over $X$. 
		\item There is an equivalence of twisting functions $\eta\sim 0$. 
	\end{enumerate}
\end{pro}

\begin{proof}
	The equivalence of (2) and (3) follows immediately from the classification result {Corollary} \ref{cor:twistings_are_PBs} and that fact that, by definition $K\times X\cong K\times_0 X$. It is clear that $x\mapsto (0,x)$ is a section of the trivial bundle $K\times X$, so (2) implies (1). 
	
	If $\pi_\eta$ admits a section then we can form a diagram 
	\[
	\begin{tikzcd}
		P\arrow[r]\arrow[d] & K\times (K\times_\eta X)\arrow[d,"\on{pr}_2"] \arrow[r,"m"] & (K\times_\eta X)\arrow[d,"\pi_\eta"] \\
		X \arrow[r,"s"]\arrow[rr,bend right=1em,"="'] & (X\times_\eta K) \arrow[r,"\pi_\eta"] & X 
	\end{tikzcd}
	\]
	by iterated pullback. The right-hand square is pullback by the freeness of the $K$-action, and the left is defined to be pullback. Since the bottom composite is the identity, $P$ is identified with $K\times_\eta X$, but by the construction of pullbacks in $\catSet$, the left-hand pullback square identifies $P$ with $K\times X$. Thus, (1) implies (2). 
\end{proof}

\subsection{The theory for a simplicial Abelian group}
\label{sec:The theory for a simplicial Abelian group}

We now specialize to the case of a simplicial Abelian group.
Let $A\in \sf{Ab}_\Delta$. We denote by 
\[
\begin{tikzcd}
	M :&[-3em] \sf{Ab}_\Delta \arrow[r,shift left] & \sf{Ch}^+\arrow[l,shift left] &[-3em] : \Gamma 
\end{tikzcd}
\]
the Dold-Kan correspondence (see, e.g., \cite[Ch. III, \S 2]{goerss2009simplicial}). Explicitly, $M$ is the \emph{normalized chains} functor which sends a simplicial Abelian group $K$ to the chain complex $M(K)$ with $n$-chains
\[
M(K)_n =\bigcap_{i=0}^{n-1} \on{ker}(d_i)\subset K_n 
\]
and differentials given by $(-1)^nd_n$. The functor $\Gamma$ sends a chain complex $C$ to a simplicial Abelian group $\Gamma(C)$ whose group of $n$-simplices is 
\[
\Gamma(C)_n=\bigoplus_{[k]\twoheadrightarrow [n]} C_k.
\]
The classic theorem of Dold and Kan says that $M$ and $\Gamma$ define an equivalence of categories. We will leverage this to connect principal bundles to cohomology groups as follows. 

\begin{pro}\label{prop:WA_DK_shift}
	There is an isomorphism of chain complexes
	\[
	M({\overline{W}A})\cong (MA)[-1]
	\] 
	natural in $A$, or, equivalently, a natural isomorphism 
	\[
	{\overline{W}A} \cong \Gamma(MA[-1]).
	\]
\end{pro}

\begin{proof}
	This is the content of \cite[Eq. 4.58]{JardineEtale}.
\end{proof}

Denote by $\underline{\sf{Ch}^+}(C_\bullet,D_\bullet)$ the chain complex of maps between the chain complexes $C_\bullet$ and $D_\bullet$. Denote by $\underline{\sf{Ab}_\Delta}(K,L)$ the simplicial Abelian group of maps between simplicial Abelian groups $K$ and $L$. 

\begin{pro}
	Let $C_\bullet$ and $D_\bullet$ be chain complexes. There is a natural equivalence of simplicial Abelian groups 
	\[
	\begin{tikzcd}
		\Gamma\left(\underline{\sf{Ch}^+}(C_\bullet, D_\bullet)\right)\arrow[r,"\simeq"] & \underline{\sf{Ab}_\Delta}(\Gamma(C_\bullet),\Gamma(D_\bullet)).
	\end{tikzcd}
	\]
\end{pro}

\begin{proof}
	This is \cite[Thm 6.3.1]{simpenrich}. 
\end{proof}

\begin{cor}\label{cor:Homotopy_is_Homology}
	There is a natural isomorphism of groups 
	\[
	\pi_0\left(\underline{\sf{Ab}_\Delta}(\Gamma(C_\bullet),\Gamma(D_\bullet))\right)\cong H_0(\underline{\sf{Ch}^+}(C_\bullet,D_\bullet)).
	\]
\end{cor}

\begin{proof}
	Under the Dold-Kan correspondence, simplicial homotopy corresponds to homology. See, e.g., \cite[Thm 8.4.1]{Weibel} for details.
\end{proof}

If we denote by $\ZZ[-]$ the free simplicial Abelian group functor
\[
\begin{tikzcd}
	{\ZZ[-]}: &[-3em] \catSet_\Delta \arrow[r] & \sf{Ab}_\Delta  
\end{tikzcd}
\]
and note that it is left adjoint to the forgetful functor $F$ from simplicial Abelian groups to simplicial sets, we obtain the following corollary. 

\begin{cor}\label{cor:Homology_homotopy_classes_WA}
	There is a natural isomorphism of groups 
	\[
	[X,{\overline{W}A}]\cong H_{-1}(\underline{\sf{Ch}^+}(M(\ZZ[X]),M(A))).
	\]
\end{cor}

\begin{proof}
	This amounts to chaining together the isomorphisms of \autoref{cor:Homotopy_is_Homology} and \autoref{prop:WA_DK_shift}, with the isomorphism 
	\[
	H_{-n}(C_\bullet)\cong H_0(C_\bullet[-n]). 
	\]	
	More explicitly, we have 
	$$
	\begin{aligned}
		[X,\barW A] &= \pi_0( \underline{\sf{Set}_\Delta}(X,\barW A) )\\ 
		&\cong \pi_0( \underline{\sf{Ab}_\Delta}(\ZZ[X],\barW A) ) \\
		&\cong \pi_0( \underline{\sf{Ab}_\Delta}(\ZZ[X],\Gamma(MA[-1])) ) \\
		&\cong H_0( \underline{\sf{Ch^+}}(M(\ZZ[X]),MA[-1]) ) \\
		&\cong H_0( \underline{\sf{Ch^+}}(M(\ZZ[X]),MA)[-1] ) \\
		&\cong H_{-1}( \underline{\sf{Ch^+}}(M(\ZZ[X]),MA)). 
	\end{aligned}
	$$	
\end{proof}

\begin{thm}\label{cor:addition_twistings_is_addition_cohomology}
	For a simplicial Abelian group $K$, the isomorphisms 
	\[
	\Twist_K(X)_{/\sim}\cong [X,\overline{W}K] \cong H_{-1}(\underline{\sf{Ch}^+}(M(\ZZ[X]),M(A)))
	\]
	of Corollaries \ref{cor:Homology_homotopy_classes_WA} and \ref{cor:twistings_are_PBs} define an isomorphism of groups, where $\Twist_K(X)_{/\sim}$ is equipped with the levelwise sum of twisting functions.
\end{thm}

\begin{proof}
	It is immediate from construction that the bijection $[X,\overline{W}K]\cong \Twist_K(X)_{/\sim}$ sends the levelwise sum of maps to the levelwise sum of twisting functions. Thus, the latter must define a group structure on $\Twist_K(X)_{/\sim}$ isomorphic to that on $[X,\overline{W}K]$. Corollary \ref{cor:Homology_homotopy_classes_WA} completes the proof. 
\end{proof}

\subsection{The case of an Abelian group} 
\label{sec:The case of an Abelian group}

We now specialize even further. Let $H\in \sf{Ab}$, and let $A=N(H)$ denote
the nerve of $H$. By explicit construction, this means that $A=\Gamma(H[-1])$, where $H$ is abusively used to denote the chain complex concentrated in degree $0$. 

\begin{pro}\label{prop:bij_2cohom_homotopy_classesWK}
	There is a natural bijection 
	\[
	[X,{\overline{W}A}]\cong H_{-2}(\underline{\Ch^+}(M(\ZZ[X]),H))=H^2(X,H).
	\]
\end{pro}

\begin{proof}
	This is simply an application of Corollary \ref{cor:Homology_homotopy_classes_WA}. Note that the additional degree shift follows because $N(H)$ corresponds under the Dold-Kan correspondence to $H[-1]$. 
\end{proof}

In this case, we can trace through the equivalences to obtain an explicit 2-cocycle representing a twisting function. Given a twisting function $\{\eta_n:X_n\to A_{n-1}\}$, the degree 2 component of the corresponding map to ${\overline{W}A}$ is 
\[
\overline{\theta^\eta_2}(x)=(\eta_2(x))\in A_1= H .
\]
Applying the normalized chain complex functor $M$, we then obtain the map 
\[
\begin{tikzcd}[row sep=0em]
	M(\overline{\theta}): &[-3em] \ker(d_0)\cap\ker(d_1) \arrow[r] & H \\
	& x \arrow[r,mapsto] & \eta_2(x). 
\end{tikzcd}
\]
To turn this into a cocycle on $X_2$ we use \cite[Chapter 2, Thm 2.1]{goerss2009simplicial}, which says that the composite 
\[
\begin{tikzcd}
	M(\ZZ[X_2]) \arrow[r] & \ZZ[X_2] \arrow[r] & \ZZ[X_2]/\ZZ[s(X_1)] 
\end{tikzcd}
\]
is an isomorphism, where $\ZZ[s(X_1)]$ is the subgroup of $\ZZ[X_2]$ generated by the degenerate 2-simplices. The inverse isomorphism is induced by the map 
\[
\begin{tikzcd}[row sep=0em]
	X_2 \arrow[r] & \ker(d_0)\cap \ker(d_1)\\
	x \arrow[r,mapsto] & x-s_0(d_0)x-s_1d_1(x-s_0d_0(x)).
\end{tikzcd}
\]
As a result, the twisting function $\eta$ is sent to the cocycle $\gamma_\eta:X_2\to H$ defined by
\[
\gamma_\eta(x)= \eta_2(x-s_0(d_0)x-s_1d_1(x-s_0d_0(x)))
\]
Notice that this cocycle $\gamma_\eta$ vanishes on degenerate simplices, that is, it is \emph{normalized}. 

On the other hand, we have 

\begin{lem}\label{lem:twist_fxn_norm_cocycle}
	Let $\gamma:X_2\to H$ be a normalized 2-cocycle on $X$. Then the maps $\eta_n:X_n\to N(H)_{n-1}$ defined inductively by 
	\[
	\eta_1(x)=0,\qquad \eta_2(x)=\gamma(x)
	\]
	and 
	\[
	\eta_n(x)=\left(\gamma(d_3\cdots d_n(x)),(\eta_{n-1}(d_1x)-\eta_{n-1}(d_0x))\right)
	\]
	defines a twisting function on $X$. 
\end{lem} 

\begin{proof}
	This follows by direct verifications of the relations between $\eta$ and the face and degeneracy maps. The details can be found in Appendix \ref{sec:simp_comp}. 
\end{proof}

\begin{pro}
	The construction of Lemma \ref{lem:twist_fxn_norm_cocycle} defines a bijection between cohomology classes of normalized 2-cocycles on $X$ with values in $H$ and equivalence classes of twisting functions. Moreover, under this bijection, the group structure on  $H^2(X,H)$, and thus the induced group structure on $\on{Prin}_{N(H)}(X)$ corresponds to the levelwise addition of twisting functions. 
\end{pro}

\begin{proof}
	This is simply a specialization of Theorem \ref{cor:addition_twistings_is_addition_cohomology} to the explicit map derived in Lemma \ref{lem:twist_fxn_norm_cocycle}. 
\end{proof}

When dealing with the nerve of an Abelian group, we can strengthen the correspondence between statements (1) and (3) in Proposition \ref{prop:triv_bundle=has_section} to a statement about cochains and cocycles. To this end, let $\gamma:X_2\to H$ be a normalized 2-cocycle, and let $\eta$ be the corresponding twisting function, as described in Lemma \ref{lem:twist_fxn_norm_cocycle}. Given a normalized 1-cochain $\alpha:X_1\to H$ such that $\partial \alpha=\gamma$, we construct a section of the corresponding (trivial) $N(H)$-bundle $\pi_\eta:N(H)\times_\eta X\to X$ as follows. 

By Lemma \ref{lem:N(H)-bundle_rel_2-cosk} in the appendix, $\pi_\eta$ is relatively 2-coskeletal over $X$. Thus, it suffices for us to construct a map of 2-truncations $\on{tr}_2(X)\to \on{tr}_2(N(H)\times_\eta X)$ over $\on{tr}_2(X)$ by Lemma \ref{lem:rel_cosk_unique_ext}. We define a map $(\varphi_\alpha)_i(x)=(\phi_i(x),x)$ for $i=0,1,2$ by 
\[
\begin{aligned}
	\phi_0(x) & =0\\
	\phi_1(x)&=\alpha(x)\\
	\phi_2(x)&=(\alpha(d_2x),\alpha(d_1x)-\alpha(d_2x)).
\end{aligned}
\]
Denote by $\on{No}(X,H,\gamma)$ the set of normalized $H$-valued 1-cochains $\alpha$ on $X$ such that $\partial \alpha=\gamma$. 

\begin{pro}\label{prop:sections_and_trivializing_cochains}
	For a twisting function $\eta$ corresponding to a normalized 2-cocycle $\gamma$, the map
	\[
	\begin{tikzcd}[row sep=0em]
		\on{No}(X,H,\gamma) \arrow[r] & \Sec(\pi_\eta)\\
		\alpha\arrow[r,mapsto] & \varphi_\alpha
	\end{tikzcd}
	\] 
	is a bijection.
\end{pro}

\begin{proof}
	We first must check that $\varphi_\alpha$ is simplicial on the 2-truncation. We first check the face maps 
	\[
	\begin{aligned}
		d_i(\phi_2(x),x)&= d_i((\alpha(d_2x),\alpha(d_1x)-\alpha(d_2x)),x)\\
		&= \begin{cases}
			(\alpha(d_1(x))-\alpha(d_2(x))+\eta_2(x),d_0x) & i=0\\
			(\alpha(d_1(x)),d_1(x)) & i=1\\
			(\alpha(d_2(x)),d_2(x)) & i=2 
		\end{cases}\\
	&= \begin{cases}
		(\alpha(d_0(x)),d_0(x)) & i=0\\
		(\alpha(d_1(x)),d_1(x)) & i=1\\
		(\alpha(d_2(x)),d_2(x)) & i=2 
	\end{cases}\\
    &=\alpha(d_i(x))
	\end{aligned}
	\] 
	as desired. Here, the third equality follows because 
	\[
	\eta(x)=\partial \alpha(x) =\alpha(d_0(x))-\alpha(d_1(x))+\alpha(d_2(x)).
	\]
	The relations with the face maps when $x\in X_1$ are immediate, since $N(H)_0$ is a singleton.
	
	We then check the relations with the degeneracies. For $x\in X_0$, 
	\[
	s_0(\phi_0(x))=s_0(0)=0=\alpha(s_0(x))
	\]
	since $\alpha$ is normalized. For $x\in X_2$, 
	\[
	\begin{aligned}
		\phi_2(s_0(x))&=(\alpha(d_2s_0(x)),\alpha(d_1s_0(x))-\alpha(d_2s_0(x)))\\
		& =(\alpha(s_0d_1(x)),\alpha(x)-\alpha(s_0d_1(x)))\\
		&=(0,\alpha(x))=s_0(\phi_2(x)) 
	\end{aligned}
	\]
	by the simplicial identities and the normalization of $\alpha$. Finally, 
	\[
	\begin{aligned}
		\phi_2(s_1(x)) &= (\alpha(d_2s_1(x)),\alpha(d_1s_1(x))-\alpha(d_2s_1(x)))\\
		& =(\alpha(x),\alpha(x)-\alpha(x))\\
		&=(\alpha(x),0)
	\end{aligned}
	\]
	as desired. This establishes that the map in question is well-defined. 
	
	On the other hand, given such a section $\varphi:X\to N(H)\times_\eta X$, expressed as $\varphi(x)=(\psi(x),x)$, we define a 1-cochain $\beta:X_1\to H$ by simply setting $\beta(x)=\psi_1(x)$. To see that this is a normalized 1-cochain, we simply note that 
	\[
	\beta(s_0(x))=\psi_1(s_0(x))=s_0(\psi_0(x))=s_0(0)=0. 
	\]    
	To see that $\partial\beta =\gamma$, note that $\psi_2(x)=(\psi_2^1(x),\psi_2^2(x))$ must satisfy the equations 
	\begin{equation}\label{eq:sys_eq_section}
	\begin{aligned}
		\psi_2^2(x)+\eta_2(x)&=\psi_1(d_0(x)) \\
		\psi_2^1(x)+\psi_2^2(x)&=\psi_1(d_1(x))\\
		\psi_2^1(x)&=\psi_1(d_2(x)).
	\end{aligned}
	\end{equation}
	Taking the alternating sum of these equations yields 
	\[
	\gamma(x)=\eta_2(x)=\psi_1(d_0(x))-\psi_1(d_1(x))+\psi_1(d_2(x))=\beta(d_0(x))-\beta(d_1(x))+\beta(d_2(x))
	\]
	as desired. 
	
	Finally, to complete the proof of bijectivity we check that $\varphi_\beta=\varphi$. However, we see that the degree 1 and 2 components agree, and it is easy to see that the equations (\ref{eq:sys_eq_section}) completely determine $\psi_2^2(x)$ and $\psi_2^1(x)$ from the values $\eta_2(x)$ and the function $\psi_1$. Thus, the proposition is proven.
\end{proof}

\subsection{The tensor product of principal bundles}
\label{sec:The tensor product of principal bundles}

We now return to the case of a generic simplicial Abelian group. Let $K$ be a simplicial Abelian group, and $X\in \catSet_\Delta$ a simplicial set. Throughout, we write $K$ multiplicatively, to better accord with the notation for group actions. Let 
\[
\begin{tikzcd}[row sep=0em]
	E \arrow[r,"\pi_E"] & X \\
	F \arrow[r,"\pi_F"] & X
\end{tikzcd}
\]
be two principal $K$-bundles over $X$. Define an equivalence relation on the set $E\times_X F$  by setting $(ke,f)\sim (e,kf)$ for all $k\in K$, $e\in E$, and $f\in F$, and let 
\[
E\otimes_K F:=(E\times_X F)_{/\sim}.
\]
Note that, equivalently, the equivalence relation is $(e,f)\sim (x,y)$ if and only if there exists $k\in K$ such that $(ke,k^{-1}f)=(x,y)$. Also note that since $E$ and $F$ are principal bundles, if such a $k$ exists, it is unique. This construction precisely mirrors that of the tensor product of principal $A$-bundles for a topological Abelian group $A$. See, for instance \cite[Example 2.2.4]{BunkePrincipal} for details.

\begin{pro}
	The canonical map $\pi_{E\otimes F}:E\otimes_K F \to X$ is a principal $K$-bundle.  
\end{pro}

\begin{proof}
	We first note that, given a map $\phi:[n]\to [m]$ in the simplex category and $(k,e,f)\in K\times (E\times_X F)$, we have 
	\[
	\begin{aligned}
		\phi^\ast(k\cdot e,f) &=(\phi^\ast(k\cdot e),\phi^\ast(f))&\\
		&=(\phi^\ast(k)\cdot\phi^\ast( e),\phi^\ast(f))\\
		&\sim(\phi^\ast( e),\phi^\ast(k)\cdot\phi^\ast(f))\\
		&=(\phi^\ast( e),\phi^\ast(k\cdot f))\\
		&= \phi^\ast(e,k\cdot f)
	\end{aligned}
	\] 
	so that the simplicial structure on $E\times_X F$ descends to $E\otimes_K F$. Moreover, since multiplication by elements of $k$ occurs fiberwise, the map $\pi_{E\otimes F}$ to $X$ remains well defined. 
	
	We then show that $E\otimes_K F$ carries a free $K$ action with quotient $X$. The action is defined by 
	\[
	k\cdot [e,f]=[ke,f]
	\]
	and is clearly simplicial. As above, it is immediate that the canonical map $\pi_{E\otimes F}:E\otimes_K F\to X$ is $K$-equivariant.
	
	To see that the action is free and transitive on fibers,  suppose we are given $[e,f],[x,y]\in E\otimes_K F$ such that $\pi_{E\otimes F}(e,f)=\pi_{E\otimes F}(x,y)$. Then choose $k\in K$ and $h\in H$ such that 
	\[
	ke=x \quad \text{and}\quad hf=y.  
	\]
	Then we see that 
	\[
	(h^{-1}ke,f)\sim (x,y)
	\]
	so the action is transitive on fibers.
	On the other hand, suppose that $h,k\in K$ such that $(he,f)\sim (x,y)$ and $(ke,f)\sim (x,y)$. Then, by the definition of the equivalence relation $\sim$, there is a unique $g\in K$ such that $(gke,g^{-1}f)=(ghe,g^{-1}f)$. Since the action of $K$ on $F$ is free, this implies $g=1$, and so $ke=he$. Since the action of $K$ on $E$ is free, this implies $k=h$, and so the action of $K$ on $E\otimes_K F$ is free. 
\end{proof}

\begin{pro}\label{pro:Addition_of_twists_tensor}
	The cohomology class in $H_{-1}(\underline{\sf{C}^+}(M(\ZZ[X]),M(K)))$ corresponding to $E\otimes_K F$ is the sum of the cohomology classes associated to $E$ and $F$, respectively.
\end{pro}

\begin{proof}
	By Theorem \ref{cor:addition_twistings_is_addition_cohomology}, it is sufficient to compute this on bundles of the form $K\times_\eta X$ for a twisting function $\eta$, since every $K$-principal bundle is isomorphic to such a one. We note that when $E=K\times_\eta X$ and $F=K\times_\rho X$, we can identify $E\otimes_{K} F$ with the set of triples of the form $(k,1,x)$ where $k\in K$ and $x\in X$. We can compute the face and degeneracy maps, which are elementwise those of $K$ and $X$, except for $d_0$. In that case, we have 
	\[
	d_0(k,1,x)=(d_0(k)\eta(x),1\rho(x),d_0(x))\sim (d_0(k)\eta(x)\rho(x),1,d_0(x)).
	\]
	So that a twisting function associated to $E\otimes_{K} F$ is $\eta\rho$. Since the group operation on twisting functions is identified with the group operation on cohomology, the proposition follows. 
\end{proof}

In the case of a general simplicial Abelian group $K$, the same reasoning as in the proposition allows us to define a canonical isomorphism 
\[
\begin{tikzcd}[row sep=0em]
	(K\times_\eta X)\otimes_K (K\times_{\xi}X )\arrow[r] & K\times_{\eta\xi} X\\
	{[(k,x),(h,x)]}\arrow[r,mapsto] & (kh,x)
\end{tikzcd}
\]
which will be of use when we discuss monoidal functors. 

Recall that we denote by $\sf{Bun}_K(X)$ the groupoid of principal $K$-bundles over $X$. 

\begin{thm}\label{thm:bun is SMC}
	The product $-\otimes_K-$ equips $\sf{Bun}_K(X)$ with the structure of a symmetric monoidal groupoid. 
\end{thm}

\begin{proof}
	This is the content of Theorem \ref{thm:Bun_SMC} in the appendix. 
\end{proof}

As it turns out, $\sf{Bun}_K(X)$ is more than simply a groupoid with a monoidal structure; it is a 2-group. See, e.g., \cite{Baez2group} for further details. 

\begin{cor}
	The category $\sf{Bun}_K(X)$ is a commutative $2$-group. 
\end{cor}

\begin{proof}
	We need only see that every object $E\to X$ has an inverse up to isomorphism under the tensor product. But since the tensor product induces the same monoid structure on $\pi_0(\sf{Bun}_K(X))$ as the bijection with $H^2(X,K)$, it is immediate that this property holds. 
\end{proof}

Since $\sf{Bun}_K(X)$ is a commutative 2-group, it can be equipped with a symmetric monoidal automorphism $\mathcal{I}:\sf{Bun}_K(X)\to \sf{Bun}_K(X)$ which sends each bundle $E\to X$ to a weak inverse $\overline{E}\to X$ under $\otimes_K$. However, if we restrict our attention to bundles given by twisted products --- i.e., restrict to the full monoidal subcategory $\sf{Bun}_K^{\on{t}}(X)$ on the twisted products $K\times_{\eta} X$ --- this duality may be given the computationally convenient form 
\[
\begin{tikzcd}
	K\times_\eta X \arrow[r,mapsto] & K\times_{\eta^{-1}} X. 
\end{tikzcd}
\]

In general, there is no \emph{bundle} isomorphism between a bundle $E\to X$ and its inverse $\overline{E}\to X$. However, as the next lemma shows, there is a \emph{simplicial} isomorphism between $K\times_\eta X$ and $K\times_{\eta^{-1}} X$.  

\begin{lem}
	Let $K$ be a simplicial Abelian group, and $\eta$ a $K$-valued twisting function on $X$. Then the map 
	\[
	\begin{tikzcd}[row sep=0em]
		f:&[-3em ]K\times_\eta X \arrow[r] & K\times_{\eta^{-1}}X \\
		&	(k,x) \arrow[r,mapsto] & (k^{-1},x)
	\end{tikzcd}
	\]
	is an isomorphism of simplicial sets over $X$. 
\end{lem}

\begin{proof}
	Since the structure maps $d_i$ for $i>0$ and $s_i$ act as homomorphisms on the $K$ component, it is immediate that they commute with $f$. It thus suffices for us to check that $f$ commutes with $d_0$. We then compute 
	\[
	f(d_0^\eta(k,x))=f(d_0(k)\eta(x),d_0(x))=((d_0(k)\eta(x))^{-1},x)=d_0^{\eta^{-1}}(k^{-1},x)=d_0^{\eta^{-1}}(f(k,x)),
	\]
	where we denote by $d^\eta_i$ the face map in $K\times_\eta X$. Since the same formula provides an inverse to $f$, the lemma is proven. 
\end{proof}

Note that the map $f$ is not $K$-equivariant, but rather equivariant along the inverse map for $K$. Thus, while it is an isomorphism of simplicial sets over $X$, it is not an isomorphism of bundles.

\section{Twisted distributions}
\label{sec:Twisted distributions}

Have now developed the requisite theory of principal bundles, we come to the main theme of this paper: distributions on principal bundles. We term these \emph{twisted distributions} since they can be viewed as variants of the simplicial distributions of \cite{okay2022simplicial} which are twisted according to a twisting function $\eta$ or the corresponding cohomology class. 

Before defining twisted distributions, however, we must establish a few preliminaries. 

\begin{defn}
	Given a semiring $R$, the \emph{$R$-distributions functor} is the functor 
	\[
	\begin{tikzcd}
		D_R: &[-3em] \catSet \arrow[r] & \catSet 
	\end{tikzcd}
	\]
	which sends $X$ to the set 
	\[
	D_R(X)=\left\lbrace p:X\to R\;\mid \;p \text{ {with finite support}, }\sum_{x\in X} p(x)=1\right\rbrace
	\]
	and sends $f: X\to Y$ to the map 
	\[
	\begin{tikzcd}[row sep=0em]
		D_R(f):&[-3em] D_R(X) \arrow[r] & D_R(Y) \\
		 & p \arrow[r,mapsto] & \left(y \mapsto {\displaystyle\sum_{x\in f^{-1}(y)}} p(x)\right).
	\end{tikzcd}
	\]
\end{defn}

Composition with $D_R$ yields a functor $\catsSet\to \catsSet$, which we also denote by $D_R$. {Following \cite{okay2022simplicial} we typically call a map of simplicial sets $X\to D_R(Y)$ a \emph{simplicial distribution}. 
{For twisted distributions we   borrow the more general notion of simplicial distributions introduced in \cite{barbosa2023bundle}.} In the special case where $R=\RR_{\geq 0}$ is the non-negative reals, we write $D$ for the functor $D_R$.

\begin{defn}
	Let $K$ be a simplicial group, and $\pi_E:E\to X$ a principal $K$-bundle. A \emph{twisted distribution} on $\pi_E$ is a simplicial map $p: X\to D_R(E)$ which makes the diagram 
	\begin{equation}\label{dia:simp-dist}
	\begin{tikzcd}
		 & D_R(E)\arrow[d,"D_R(\pi_E)"] \\
		 X\arrow[r,"\delta"']\arrow[ur,"p"] & D_R(X)
	\end{tikzcd}
	\end{equation}
	commute. 
	
	We write $\sDist(\pi_E)$ for the {set} of simplicial 
	distributions on $\pi_E$. 
	A section $\varphi:X\to E$ of the map $\pi_E:E\to X$ gives a simplicial distribution defined by $\delta^\varphi:X\xrightarrow{\varphi} E\xrightarrow{\delta} D_R(E)$  called a \emph{deterministic distribution}. 
	We write $\delta:\Sec(\pi_E) \to \sDist(\pi_E)$ for the map that sends a section $\varphi$ to the associated deterministic distribution $\delta^\varphi$.
\end{defn}

A simplicial distribution on the principal bundle $\pi_\eta:K\times_\eta X \to X$   is called an {\it $\eta$-twisted distribution}.
For a simplicial set map $f:X\to Y$ it will be convenient to write $f_x=f_n(x)$. With this notation an $\eta$-twisted distribution, i.e., a simplicial set map $p:X\to D_R(E)$ making Diagram (\ref{dia:simp-dist}) commute, consists of   
a family of distributions $\set{p_x:\, x\in X_n,\, n\geq 0}$, where $p_x\in D_R(K_n\times X_n)$ with support contained in $K_n\times \set{x}=\set{(k,x):\, k\in K_n}$, such that for every $\theta:[m]\to [n]$ and $x\in X_n$ the following equation holds
$$
D_R(\theta^*)(p_{x}) = p_{\theta^*(x)}.
$$
When $K=NH$ we adopt the notation  $p_x^{h_1h_2\cdots h_n}=p_n(x)(h_1,h_2,\cdots,h_n)$. By Lemma \ref{lem:twist_fxn_norm_cocycle} a twisted bundle is determined by a normalized cocycle $\gamma: X_2 \to H$. In this case we will write $X_\gamma$ for the twisted product $K\times_\eta X$.

\begin{ex}\label{ex:Delta-2-3} 
Let us consider simplicial distributions on $\Delta^2$ twisted by a normalized cocycle $\gamma:(\Delta^2)_2\to \ZZ_2$. A twisted simplicial distribution is given by a commutative diagram
	$$
	\begin{tikzcd}
		& D_R(\Delta^2_\gamma) \arrow[d,"D_R(\pi)"]  \\
		\Delta^2 \arrow[r,"\delta"'] \arrow[ru,"p"] & D_R(\Delta^2) 
	\end{tikzcd}
	$$
Let $\sigma$ denote the non-degenerate $2$-simplex of $\Delta^2$.
Then $p_\sigma$ is a distribution on $\ZZ_2^2\times (\Delta^2)_2$ with support on $\set{(ab,\sigma):\, a,b,\in \ZZ_2}$. The $d_0$-face map is twisted:
	$$
	p_{d_0\sigma}^a =
	p_\sigma^{0(a+\alpha)}+p_\sigma^{1(a+\alpha)} 
	$$
	where $\alpha=\eta_2(\sigma)=\gamma(\sigma)$ by Lemma \ref{lem:twist_fxn_norm_cocycle}.
	
	Next, consider $\Delta^3$ and a normalized $2$-cochain $\gamma$ on this simplicial set. Let $\sigma$ denote the non-degenerate $3$-simplex of $\Delta^3$.
The $d_0$-face map is given by
	$$
	p_{d_0\sigma}^{ab} = p_\sigma^{0(a+\alpha)(b+\beta)} + p_\sigma^{1(a+\alpha)(b+\beta)}
	$$
	where $(\alpha,\beta)=\eta_3(\sigma)=(\gamma(d_3\sigma),\gamma(d_1\sigma)-\gamma(d_0\sigma))$ by Lemma \ref{lem:twist_fxn_norm_cocycle}.
\end{ex}

\subsection{The category of $K$-bundle scenarios}

We now turn our attention to constructing appropriate categories of $K$-bundle scenarios, which we will accomplish with the help of the Grothendieck construction. The reader looking for more background on 2-categories and the Grothendieck construction can consult, e.g., Chapters 2 and 10 of \cite{Johnson_Yau_2021}.

Let $K$ be a simplicial group. Given a simplicial map $f:X\to Y$, pullback along $f$ induces a functor 
\[
\begin{tikzcd}
	f^\ast: &[-3em] \sf{Bun}_K(Y) \arrow[r] & \sf{Bun}_K(X). 
\end{tikzcd}
\]
The universal property of the pullback provides natural isomorphisms 
\[
\xi_{f,g}:g^\ast \circ f^\ast\cong (f\circ g)^\ast 
\]
for any composable simplicial maps $f$ and $g$. Similarly, the universal property provides a natural isomorphism 
\[
(\on{id}_X)^\ast \cong \on{Id}_{\sf{Bun}_K(X)}. 
\]
These data provide a pseudofunctor 
\[
\begin{tikzcd}[row sep=0em]
	\sf{Bun}_K: &[-3em] \catsSet^\op \arrow[r] & \sf{Grpd}\\
	& X \arrow[r,mapsto] & \sf{Bun}_K(X). 
\end{tikzcd}
\]

\begin{defn}
	The \emph{category of $K$-bundle scenarios} $\catbScen_K$ is the covariant Grothendieck construction of the functor $\sf{Bun}_K$. 
\end{defn}

This category was first introduced in \cite{barbosa2023bundle}.
More explicitly, we can describe $\catbScen_K$ as the category whose objects are principal $K$-bundles $\pi_E:E\to X$, and whose morphisms from $\pi_E:E\to X$ to $\pi_F:F\to Y$ are given by commutative diagrams 
\[
\begin{tikzcd}
	E \arrow[d,"\pi_E"'] & f^\ast E\arrow[l]\arrow[d]\arrow[r,"\alpha"] & F\arrow[d,"\pi_F"]\\
	X & Y\arrow[l,"f"] \arrow[r,"="'] & Y 
\end{tikzcd}
\]
The composition of this morphism with a morphism 
\[
\begin{tikzcd}
	F \arrow[d,"\pi_F"'] & g^\ast F\arrow[l]\arrow[d]\arrow[r,"\beta"] & B\arrow[d,"\pi_B"]\\
	Y & Z\arrow[l,"g"] \arrow[r,"="'] & Z
\end{tikzcd}
\]
is the morphism
\[
\begin{tikzcd}[column sep=5em]
	E \arrow[d,"\pi_E"'] & (f\circ g)^\ast E\arrow[l]\arrow[d]\arrow[r,"\beta \circ g^\ast(\alpha)\circ \xi_{f,g}"] & B\arrow[d,"\pi_B"]\\
	X & Z\arrow[l,"f\circ g"] \arrow[r,"="'] & Z
\end{tikzcd}
\]

We then construct a lax natural transformation 
\[
\begin{tikzcd}
	\sDist: &[-3em] \sf{Bun}_K \arrow[r, Rightarrow] & \underline{\catSet} 
\end{tikzcd}
\]
where the latter denotes the constant functor $\catsSet^\op\to \Cat$  on the category $\catSet$\footnote{The manipulations we use here 
based on
Grothendieck constructions involve some size issues --- for example, the category $\catSet$ is not small, and so to consider it as an object of $\Cat$ would require additional axiomatics. However, the categories and functors we construct using the Grothendieck construction can be given explicit definitions which avoid these issues, and the reader not willing to resort to alternate foundations may instead resort to these.}. The components are the functors 
\[
\begin{tikzcd}
	\sDist_X: &[-3em] \sf{Bun}_K(X) \arrow[r] & \catSet 
\end{tikzcd}
\]
which send bundle $\pi_E: E\to X$ to $\sDist(\pi_E)$ and send a bundle isomorphism to the corresponding isomorphism of sets of simplicial distribution. To see lax naturality, note that for $f:X\to Y$ a simplicial map and $\pi_E:E\to Y$ a $K$-bundle, we can define a map 
\[
\begin{tikzcd}[row sep=0em]
	\psi_E^f:&[-3em] \sDist_Y(\pi_E)\arrow[r] & (\sDist_X\circ f^\ast )(\pi_E) 
\end{tikzcd}
\]
which sends a distribution $p$ on $\pi_E$ to the composition 
\[
\begin{tikzcd}
	X \arrow[r,"p\circ f\times \delta_X"] & D_R(E)\times D_R(X)\arrow[r,"m"] & D_R(E\times X)
\end{tikzcd}
\] 
Which has image in the simplicial subset $D_R(f^\ast E)$ by \cite[\S 4.2]{barbosa2023bundle}. The fact that the maps $\psi_E^f$ do, in fact, produce a lax natural transformation follows from Lemmas 4.13 and 4.14 of \cite{barbosa2023bundle}.

Since the Grothendieck construction on a pseudofunctor computes its lax colimit (see, e.g., \cite[Thm 7.4]{gepnerhaugsengnikolaus}) this lax transformation immediately yields the following proposition. 

\begin{pro}\label{prop:functoriality_bundle_scenarios}
	There is a functor 
	\[
	\begin{tikzcd}
		\sDist: &[-3em]\catbScen_K \arrow[r] & \catSet 
	\end{tikzcd}
	\]
	which sends each bundle $\pi_E:E\to X$ to the set $\sDist(\pi_E)$ of twisted distributions on that bundle, each bundle isomorphism $\alpha$ to the corresponding isomorphism on sets of twisted distributions, and each pullback square 
	\[
	\begin{tikzcd}
		E \arrow[d,"\pi_E"'] & f^\ast E\arrow[d]\arrow[l] \\
		Y & X\arrow[l,"f"]
	\end{tikzcd}
	\]  
	to the map $\psi_E^f$ on sets of twisted distributions. 
\end{pro}

\subsection{The colimit decomposition}

If the base space $X$ of a fibration $\pi:E\to X$ can be expressed as a colimit, it is possible to express the space of simplicial distributions in terms of simplicial distributions on the pullback. To make this precise, we consider a functor 
\[
\begin{tikzcd}[row sep=0em]
	\sDist_\pi: &[-3em] \left(\catSet_\Delta\right)_{/X}^\op \arrow[r] & \bf{Set} \\
	& (f:Y\to X) \arrow[r,mapsto] & \sDist\left(Y\times_X E \to Y\right).
\end{tikzcd}
\]
By  
Proposition \ref{prop:functoriality_bundle_scenarios}
any pullback diagram 
\[
\begin{tikzcd}
	E \arrow[d,"\pi_X"'] & F\arrow[d,"\pi_Y"]\arrow[l]\\
	X  & Y\arrow[l]
\end{tikzcd}
\]
yields a map 
\[
\begin{tikzcd}
	\sDist(\pi_X) \arrow[r] & \sDist(\pi_Y)
\end{tikzcd}
\]
providing the functoriality of $\sDist_\pi$. 

\begin{pro}\label{prop:Sdistpi_pres_lim}
	Given a map $\pi:E\to X$, the associated functor 
	\[
	\begin{tikzcd}
		({\catSet}_\Delta)_{/X}^\op \arrow[r,"\on{sDist}_\pi"] & {\catSet } 
	\end{tikzcd} 
	\]
	preserves limits. That is, given a diagram $F:I\to (\catSet_\Delta)_{/X}$ with colimit $Y\to X$, the natural map 
	\[
	\begin{tikzcd}
		\on{sDist}_\pi(Y\times_X E\to Y) \arrow[r] & \lim \on{sDist}_\pi(F(i)\times_X E\to F(i)) 
	\end{tikzcd}
	\]
	is a bijection. 
\end{pro}

\begin{proof}
	We fix some notation before beginning. We write 
	\[
	\begin{tikzcd}
		Y_i \arrow[dr,"\alpha_i"]\arrow[dd,"f_u"'] & \\
		& X \\
		Y_j \arrow[ur,"\alpha_j"] 
	\end{tikzcd}
	\]
	for the image under $F$ of  a morphism $u:i\to j$ in $I$. We write $\alpha:Y\to X$ for the colimit of $F$, and we write $\eta_i: Y_i\to Y$ for the colimiting cocone. We write $\pi_i:E_i:=Y_i\times_X E\to Y_i$, $\pi_Y:E_Y:=Y\times_X E\to Y$, and use an overline to denote the induced maps of total spaces, e.g. 
	\[
	\begin{tikzcd}
		E_i\arrow[d,"\pi_i"'] \arrow[r,"\overline{f}_u"] & E_j\arrow[d,"\pi_j"] \\
		Y_i \arrow[r,"f_u"'] & Y_j 
	\end{tikzcd}
	\]
	
	Given a compatible collection $\{p_i: Y_i \to D_R(E_i)\}_{i\in I}$ of simplicial distributions, we wish to show that there is a unique simplicial distribution $p:Y\to D_R(E_Y)$ such that $(\eta_i)_\ast(p)=p_i$ for every $i\in I$. We first show uniqueness. Suppose that such a distribution $p$ exists, and let $(e,y)\in (E_Y)_n$. Since $Y$ is the colimit of the diagram $F$, this means that there is an $i\in I$ and $y_i\in (Y_i)_n$ such that $\eta_i(y_i)=y$. Thus, we see that desired compatibility condition implies 
	\[
	p_i(y_i)=p(\eta_i(y_i))\cdot \delta_{Y_i}(y_i)
	\] 
	applying this to the element $(e,y_i)$, we see that this implies 
	\[
	p_i(y_i)(e,y_i)=p(y)(e,y)\cdot \delta_{y_i,y_i}=p(y)(e,y).
	\]
	This immediately puts paid to uniqueness,as $p$ is specified at every $y\in Y$ and $(e,y)\in E_Y$ by the formula 
	\[
	p(y)(e,y)=p(y_i)(e,y_i).
	\]
	We then must check that this formula yields a well-defined simplicial distribution. 
	
	To see well-definedness, let $u:i\to j$ in $I$, and let $y_i\in (Y_i)_n$ and $y_j\in (Y_j)_n$ such that
	\[
	f_u(y_i) =y_j \quad \text{and}\quad \eta_j(y_j)= y.
	\] 
	Further let $(e,y)\in E_Y$. The compatibility condition for the collection $\{p_i:Y_i\to D_R(E_i)\}_{i\in I}$ is precisely that
	\[
	(f_u)_\ast(p_j)=p_i 
	\]
	and so we see that 
	\[
	p_i(y_i)(e,y_j) =p_j(y_j)(e,y_j) 
	\]
	so that the two definitions of $p(e,y)$ coincide. Thus, the formula above determines a well-defined value $p(y)(e,y)$ for every $(e,y)\in E_Y$. Extending by $0$ determines $p(y)(e,z)$ for every $z\neq y$. 
	
	We next show that this is, indeed, a distribution. For $y\in Y$, $p(y):E\to R$ is by definition supported on the fiber $\pi_Y^{-1}(y)$. Given $y_i\in Y_i$ such that $\eta_i(y_i)=y$, this canonically identifies the fibers $\pi_i^{-1}(y_i)$ and $\pi_Y^{-1}(y)$. Thus, we see that 
	\[
	\sum_{(e,z)\in(E_Y)_n} p(y)(e,z) =\sum_{(e,y)\in \pi_Y^{-1}(y)} p(y)(e,y)=\sum_{(e,y_i)\in \pi_i^{-1}(y_i)} p(y_i)(e,y_i)=1  
	\] 
	and, since the right-hand side is equal to $1$, we see that $p(y)$ is a distribution. Moreover, by construction, $p(y)$ lies over $\delta_Y(y)$. 
	
	We now need only check that the map 
	\[
	\begin{tikzcd}[row sep=0em ]
		p: &[-3em] Y \arrow[r] & D_R(E) \\
		& y \arrow[r,mapsto] & p(y)
	\end{tikzcd}
	\]
	respects the simplicial structure. Let $\theta:[n]\to [m]$ in $\Delta$, and suppose that $y=\theta^\ast(z)$ for $z\in (Y)_m$. Choosing $z_i\in (Y_i)_m$ such that $\eta_i(z_i)=z$, the fact that $\eta_i$ is a simplicial map means that we can define $y_i=\theta^\ast(z_i)$ and have $\eta_i(y_i)=y$. We then have 
	\[
	p(y)=p(\theta^\ast(z))=p_i(\theta^\ast(z_i))=D_R(\theta^\ast)(p_i(z)).
	\]
	But $D_R(\theta^\ast)(p_i(z_i))$ is defined on $(e,x_i)\in (E_i)_n$ by  
	\[
	\begin{aligned}
		D_R(\theta^\ast)(p_i(z_i))(e,x_i) &=\sum_{(g,w_i)\in (\theta^\ast)^{-1}(e,x_i)} p_i(z_i)(g,w_i)\\
		&=\sum_{ (g,z_i)\in (\theta^\ast)^{-1}(e,x_i)} p_i(z_i)(g,z_i)\\
		&=\sum_{(g,z)\in (\theta^\ast)^{-1}(e,x)}p(z)(g,z)\\
		&=\sum_{(g,w)\in (\theta^\ast)^{-1}(e,x)}p(z)(g,w)\\
		&=D_R(\theta^\ast)(p(z))(e,\eta_i(x_i)).
	\end{aligned}
	\]
	Thus $p(\theta^\ast(z))=D_R(\theta^\ast)(p(z))$, and so $p$ is a simplicial map, completing the proof. 
\end{proof}

\subsection{The polytope of twisted distributions}

We can apply 
Proposition \ref{prop:Sdistpi_pres_lim}
to the canonical representation of $X$ as a colimit over its category of simplices to obtain a useful characterization of simplicial distributions on $\pi$. 

\begin{defn}
	Given $\pi:E\to X$ a principal $K$-bundle, we define two functors 
	\[
	\begin{tikzcd}
		D_\pi: &[-3em] (\Delta_{/X})^\op \arrow[r,hookrightarrow] & (\catSet_\Delta)_{/X}^\op \arrow[r,"\on{sDist}_\pi"] & \catSet 
	\end{tikzcd}
	\]
	and 
	\[
	\begin{tikzcd}
		(\Delta_{/X})^\op \arrow[r,"F_{\pi}"] & \catSet \arrow[r,"D_R"] & \catSet
	\end{tikzcd}
	\]
	where $F_\pi$ is the functor which sends $\sigma:\Delta^n\to X$ to the fiber $\pi_n^{-1}(\sigma)$ of $\pi$ over $\sigma$, which is in bijection with $K_n$. 
\end{defn}

\begin{lem}\label{lem:nat_bij_DR_FP}
	For a  principal $K$-bundle $\pi: E\to X$, there is a natural isomorphism 
	\[
	D_\pi \cong D_R\circ F_\pi.
	\]
\end{lem}

\begin{proof}
	Given $u:\Delta^n\to X$, we construct a bijection 
	\[
	\begin{tikzcd}
		\alpha_u: &[-3em] D_\pi(u) \arrow[r] & D_R(F_\pi(u)) 
	\end{tikzcd}
	\]
	as follows. Maps of simplicial sets $p:\Delta^n\to D_R(\Delta^n\times_X E)$ are in bijection with $n$-simplices of $D_R(\Delta^n\times_X E)$, which are in turn elements of  
	\[
	D_R((\Delta^n\times_X E)_n).
	\]
	This bijection is given explicitly by sending $p$ to $p_{u(\sigma)}$, i.e., evaluating on the unique non-degenerate $n$-simplex. 
	
	Simplicial distibutions are precisely those $p$ for which the diagram 
	\[
	\begin{tikzcd}
		& D_R(E\times_X \Delta^n)\arrow[d,"D_R(\pi)"]\\
		\Delta^n\arrow[r,"\delta_{\Delta^n}"']\arrow[ur,"p"] & D_R(\Delta^n)
	\end{tikzcd}
	\]
	commutes, i.e., such that the support of $p_{u(\sigma)}$ lies in $\pi_n^{-1}(u)$, i.e., $D_R(F_\pi(u))$. The naturality of this bijection is immediate from the definitions.
\end{proof}

From this lemma, we immediately obtain a characterization of simplicial distributions as a corollary. 

\begin{cor}\label{cor:twiste dist is a polytope}
	Let $X$ be a simplicial set, and $\pi:E\to X$ a principal $K$-bundle. The canonical map 
	\[
	\begin{tikzcd}
		\on{sDist}(\pi) \arrow[r] & \lim_{(\Delta_{/X})^\op} D_R\circ F_\pi
	\end{tikzcd}
	\]
	is a bijection.	
	In particular, if $X$ has finitely many non-degenerate simplices and $R=\mathbb{R}_{\geq 0}$, then $\on{sDist}(\pi)$ is a convex polytope with finitely many vertices. 
\end{cor}

\begin{proof}
	Since every simplicial set $X$ is a colimit over its category of simplices, 
	Proposition \ref{prop:Sdistpi_pres_lim} shows that the canonical map 
	\[
	\begin{tikzcd}
		\on{sDist}(\pi) \arrow[r] & \lim_{\Delta_{/X}^\op} D_\pi  
	\end{tikzcd}
	\]
	is an isomorphism. By Lemma \ref{lem:nat_bij_DR_FP}, however, we see that 
	\[
	\lim_{\Delta_{/X}^\op} D_\pi \cong \lim_{\Delta_{/X}^\op} D_R\circ F_\pi.
	\]
		When $X$ is a finite simplicial set and $R=\RR_{\geq 0}$, Lemma \ref{lem:fin_sset_cat_simps_cofinal} in the appendix shows that we may, instead, compute the limit  as the limit over $\left((\Delta_{/X})^{\leq n}\right)^\op$. Since this category is finite, we may express the limit as an equalizer of finite products:
	\[
	\begin{tikzcd}
		\sDist(\pi) \arrow[r] & \prod\limits_{\sigma} D(F_\pi(\sigma)) \arrow[r,shift left,"s"]\arrow[r,shift right,"t"'] & \prod\limits_{\psi:\sigma\to \tau} D(F_\pi(\tau))
	\end{tikzcd}
	\]
	where the former product is indexed over objects of $\left((\Delta_{/X})^{\leq n}\right)^\op$, and the latter is indexed over morphisms of $\left((\Delta_{/X})^{\leq n}\right)^\op$. Since each of the $D(F_\pi(\sigma))$ are polytopes with finitely many vertices, this means that $\sDist(\pi)$ is a subspace of a finite product of such polytopes, specified by finitely many linear equations. Thus, $\sDist(\pi)$ is a convex polytope with finitely many vertices, as desired.
\end{proof}

\subsection{The convolution product}

We now make use of the monoidal structure on principal bundles to define a global structure on simplicial distributions on $K$-bundles.  To do this, we first note that, for any simplicial sets $X,Y,Z\in \catsSet$, and maps $f:X\to Z$, $g:Y\to Z$, we have a canonical map 
\[
\begin{tikzcd}[row sep=0em]
	D_R(X)\times D_R(Y) \arrow[r,"m"] & D_R(X\times Y) \\
	(p,q) \arrow[r,mapsto] & \left((x,y)\mapsto p(x)q(y)\right).
\end{tikzcd}
\]
Let $K$ be a simplicial Abelian group. For $K$-bundles $\pi_E:E\to X$ and $\pi_F:F\to X$ form the composite
\[
\begin{tikzcd}
	\sDist(\pi_E)\times \sDist(\pi_F) \arrow[r]& \sDist(X,E)\times \sDist(X,F) \arrow[r,"m"] & \sDist(X,E\times F).
\end{tikzcd}
\]
However, for any $(p,q)\in 	\sDist(\pi_E)\times \sDist(\pi_F)$, $p(x)$ and $p(x)$ are supported on $\pi_E^{-1}(x)$ and $\pi_F^{-1}(x)$, respectively. Thus, $m(p,q)(x)$ is supported on $\pi_E^{-1}(x)\times\pi_F^{-1}(x)$. As a result, we see that we obtain a map, abusively also denoted by $m$, 
\[
\begin{tikzcd}
	m: &[-3em] \sDist(\pi_E)\times \sDist(\pi_F) \arrow[r] & \sDist(E\times_X F\to X). 
\end{tikzcd}
\]
The quotient map $E\times_X F\to E\otimes_K F$ induces a map on distributions. Composing this with $m$, we obtain a map
\[
\begin{tikzcd}
	\mu: &[-3em] \sDist(\pi_E)\times \sDist(\pi_F) \arrow[r] & \sDist(E\otimes_K F\to X). 
\end{tikzcd}
\]
\begin{defn}
	The \emph{convolution product of twisted distributions} is the map 
	\[
	\begin{tikzcd}
		\mu: &[-3em] \sDist(\pi_E)\times \sDist(\pi_F) \arrow[r] & \sDist(E\otimes_K F\to X). 
	\end{tikzcd}
	\]
	We will often denote $\mu(p,q)$ by $p\ast q$. 
\end{defn}

We also note that, for the trivial bundle $K\times X\to X$, there is a canonical choice of section $s_0:X\to K\times X$, which sends $x\mapsto (0,x)$.  We denote by $\varepsilon:\ast \to \sDist(K\times X\to X)$ the map which sends $\ast$ to the deterministic distribution on $s_0$. 

\begin{pro}\label{pro:sDist is lax SMF}
	The maps $\mu$ and $\varepsilon$ equip the functor
	\[
	\begin{tikzcd}
		\sDist: &[-3em] \sf{Bun}_K(X) \arrow[r] & \catSet 
	\end{tikzcd}
	\]
	with the structure of a lax symmetric monoidal functor. 
\end{pro}

\begin{proof}
	This is Proposition \ref{prop:Bun_sym_mon_funct_Set} in the appendix. 
\end{proof}

Applying the monoidal Grothendieck construction, we obtain 

\begin{cor}\label{cor:Grothendieck construction}
	The Grothendieck construction $\int_{\sf{Bun}_K(X)}\sDist$ is equipped with a symmetric monoidal structure given by $(\pi_E,p)\boxtimes (\pi_F,Q)=(E\otimes_K F,p\ast q)$ with unit $(K\times X,\delta^{s_0})$. Thus, $\int_{\sf{Bun}_K(X)}\sDist$ is a symmetric monoidal groupoid. 
\end{cor}

\begin{proof}
	This 
	 follows from \cite[Thm 3.13]{moellervasilakopoulou}. 
\end{proof}

\begin{rem}
{In \cite{Omongroth} and \cite{HOSconvex}, the authors of the present paper generalize the monoidal Grothendieck construction of \cite{moellervasilakopoulou} to take as input generalizations of monoidal functors into the category of convex sets. From this generalization, the category $\int_{\sf{Bun}_K(X)}\sDist$ inherits additional convex structure defined fiberwise. In particular, this structure provides the monoid of Corollary \ref{cor:convmon_twistconv} below a partially-defined convex structure, which generalizes the convex monoids of distributions from \cite{kharoof2022simplicial}.}
\end{rem}

Denote by $\sf{Bun}_K^0(X)$ the groupoid of trivial principal $K$-bundles over $X$. Then we can note that, since a tensor product of trivial bundles is trivial, $\sf{Bun}_K^0(X)$ is a symmetric monoidal subcategory of $\sf{Bun}_K(X)$. We can define another functor
\[
\begin{tikzcd}[row sep=0em]
	\on{Cl}:&[-3em]\sf{Bun}_K^0(X) \arrow[r] & \catSet\\
	  & \pi_E \arrow[r,mapsto] & D_R(\Sec(\pi_E)).
\end{tikzcd}
\]
We then define structure maps to equip $\on{Cl}$ with the structure of a lax monoidal functor. Firstly, given sections $\phi$ of $\pi_E$ and $\psi$ of $\pi_F$, we obtain a canonical map
\[
\begin{tikzcd}
	(\phi,\psi):&[-3em] X \arrow[r] & E\times_X F 
\end{tikzcd}
\]
which is a section of $E\times_X F\to X$. composing this with the quotient map, we obtain a section, which we abusively write 
\[
\begin{tikzcd}
	\phi+\psi:&[-3em] X \arrow[r] & E\otimes_K F.
\end{tikzcd}
\]
This construction provides a map
\[
\begin{tikzcd}
	\overline{\ell}: &[-3em] \Sec(\pi_E)\times \Sec(\pi_F)\arrow[r] & \Sec(\pi_{E\otimes F}). 
\end{tikzcd}
\]

\begin{lem}
	For $K$-valued twisting functions $\eta,\nu$ on $X$, and sections $\phi(x)=(s(x),x)$ and $\psi(x)=(t(x),x)$ of $\pi_\eta$ and $\pi_\nu$, respectively, the identification of $(K\times_\eta X)\otimes_K (K\times_\nu X)$ with $K\times_{\eta+\nu} X$ from the proof of Proposition \ref{pro:Addition_of_twists_tensor} identifies $\phi+\psi$ with the section 
	\[
	\begin{tikzcd}[row sep=0em]
		X \arrow[r] & K\times_{\eta+\nu}X \\
		x \arrow[r,mapsto] & (s(x)+t(x),x).
	\end{tikzcd}
	\]
\end{lem}
\begin{proof}
	Follows immediately from unwinding the definitions. 
\end{proof}

Applying $D_R$, we then obtain a composite map
\[
\begin{tikzcd}
	\ell:&[-3em]D_R(\Sec(\pi_E))\times D_R(\Sec(\pi_F))\arrow[r,"m"] &D_R(\Sec(\pi_E)\times \Sec(\pi_F)) \arrow[r,"D_R(\overline{\ell})"] & D_R(\Sec(\pi_{E\otimes F})).
\end{tikzcd}
\]
Similarly, applying $D_R$ to the inclusion $\ast \to \Sec(K\times X)$ of the zero section yields a map
\[
\begin{tikzcd}
	\ast \cong D_R(\ast) \arrow[r,"u"] & D_R(\Sec(K\times X)).  
\end{tikzcd}
\]

\begin{pro}
	The maps $\ell$ and $u$ equip $\on{Cl}$ with the structure of a lax monoidal functor. 
\end{pro}

\begin{proof}
	This is Proposition \ref{pro:Cl_lax_mon} in the appendix.
\end{proof}

For a trivial $K$-bundle $\pi_E:E\to X$, we now define a map
\[
\begin{tikzcd}
	\Theta_E: &[-3em] D_R(\Sec(\pi_E)) \arrow[r] & \sDist(\pi_E)
\end{tikzcd}
\]
as follows. Given a distribution $p$ on sections of $\pi_E$, we define the distribution $\Theta_E(p)$ by 
\[
\Theta_E(p)_x(e)=\sum_{\phi\in \Sec(\pi_E)} p(\phi) \delta_{\phi(x),e}. 
\]
To see that this map is well-defined, we must show that $\Theta(p)_x$ is a distribution, and $\Theta(p)$ is simplicial. The first follows because 
\[
\begin{aligned}
	\sum_{e\in\phi_E^{-1}(x)} \Theta(p)_x(e)&=	\sum_{e\in\phi_E^{-1}(x)} \sum_{\phi\in \Sec(\pi_E)} p(\phi) \delta_{\phi(x),e}\\
	 &= \sum_{\phi\in \Sec(\pi_E)}\sum_{e\in\phi_E^{-1}(x)}  p(\phi) \delta_{\phi(x),e}\\
	 & =\sum_{\phi\in \Sec(\pi_E)} p(\phi)=1.
\end{aligned}
\] 
On the other hand, given a simplicial map $\theta:[n]\to [m]$ in $\Delta$, we compute 
\[
\begin{aligned}
	\theta^\ast(\Theta(p)_x)(e)&=\sum_{\overline{e}\in (\theta^\ast)^{-1}(e)} \Theta(p)_x(\overline{e}) \\
	&=\sum_{\overline{e}\in (\theta^\ast)^{-1}(e)} \sum_{\phi\in \Sec(\pi_E)} p(\phi) \delta_{\phi(x),\overline{e}} \\
	&=\sum_{\phi\in \Sec(\pi_E)}\sum_{\overline{e}\in (\theta^\ast)^{-1}(e)}  p(\phi) \delta_{\phi(x),\overline{e}}\\
	&= \sum_{\phi\in \Sec(\pi_E)}  p(\phi) \delta_{\phi(\theta^\ast(x)),e}\\
	&= \Theta(p)_{\theta^\ast(x)}(e).
\end{aligned}
\]

\begin{lem}
	The maps $\Theta_E$ define a monoidal natural transformation from $\on{Cl}$ to $\sDist|_{\sf{Bun}_K^0(X)}$. 
\end{lem}

\begin{proof}
	Naturality in bundle isomorphisms is immediate. It thus remains to check monoidalness of the transformation, i.e., the commutativity of the squares 
	\[
	\begin{tikzcd}[column sep=5em, row sep=5em]
		D_R(\Sec(\pi_E))\times D_R(\Sec(\pi_F)) \arrow[r,"\Theta_E\times \Theta_F"]\arrow[d,"\ell"'] & \sDist(\pi_E)\times \sDist(\pi_F) \arrow[d,"\mu"]\\ D_R(\Sec(\pi_{E\otimes F})) \arrow[r,"\Theta_{E\otimes F}"']& \sDist(\pi_{E\otimes F})
	\end{tikzcd}
	\] 
	Starting with a pair of distributions $p,q\in D_R(\Sec(\pi_E))\times D_R(\Sec(\pi_F))$, applying the definitions yields 
	\[
	\Theta_{E\otimes F}(\ell(p,q))_x([e,f])=\sum_{\phi\in \Sec(\pi_{E\otimes F})} \left(\sum_{\psi_E,\psi_F \text{s.t.} \overline{\ell}(\psi_E,\psi_F)=\phi} p(\psi_E)q(\psi_F)\right)\delta_{\phi(x),[e,f]} .
	\]
	Since the internal sum collapses, we can rewrite this as 
	\[
	\Theta_{E\otimes F}(\ell(p,q))_x([e,f])=\sum_{(\psi_E,\psi_F)\in\Sec(\pi_E)\times \Sec(\pi_F)} p(\psi_E)q(\psi_F)\delta_{\overline{\ell}(\psi_E(x),\psi_F(x)),[e,f]} .
	\] 
	On the other hand, applying the definitions also yields 
	\[
	\mu(\Theta_E(p),\Theta_F(q))_x([e,f])= \sum_{\overline{e},\overline{f}\text{ s.t. } [\overline{e},\overline{f}]=[e,f]} \sum_{(\psi_E,\psi_F)\in\Sec(\pi_E)\times \Sec(\pi_F)} p(\psi_E)q(\psi_F)\delta_{\psi_E(x),\overline{e}} \delta_{\psi_F(x),\overline{f}} .
	\] 
	However, $\overline{\ell}(\psi_E,\psi_F)(x)=[e,f]$ if and only if there exist unique $\overline{e},\overline{f}$ such that $[\overline{e},\overline{f}]=[e,f]$ and $\psi_E(x)=\overline{e}$ and $\psi_F(x)=\overline{f}$. Thus, we can rewrite this sum as 
	\[
	\mu(\Theta_E(p),\Theta_F(q))_x([e,f])=\sum_{(\psi_E,\psi_F)\in\Sec(\pi_E)\times \Sec(\pi_F)} p(\psi_E)q(\psi_F)\delta_{\overline{\ell}(\psi_E,\psi_F)(x),[e,f]} 
	\]
	which shows that the square commutes. It is definitional that $\Theta_{K\times X}\circ u=\epsilon$, completing the proof.  
\end{proof}

\begin{cor}\label{cor:Theta_monoidal_functor}
	There is a monoidal functor 
	\[
	\begin{tikzcd}
		\Theta: &[-3em] {\displaystyle\int_{\sf{Bun}^0_K(X)}} \on{Cl} \arrow[r] & {\displaystyle \int_{\sf{Bun}^0_K(X)} \sDist}. 
	\end{tikzcd}
	\]
\end{cor}

\begin{proof}
	This follows from \cite[Thm 3.13]{moellervasilakopoulou}.
\end{proof}

If we view the (commutative) monoid $\Twist_K(X)$ of twisting functions on $X$ as a discrete symmetric monoidal category, we have the following result. 

\begin{pro}
	The functor 
	\[
	\begin{tikzcd}[row sep=0em]
		\Twist_K(X) \arrow[r] & \sf{Bun}_K(X) \\
		\eta \arrow[r,mapsto] & K\times_\eta X
	\end{tikzcd}
	\]
	is symmetric monoidal. 
\end{pro}

\begin{proof}
	This is Proposition \ref{prop:Twist_sym_mon_funct_bun} in the appendix. 
\end{proof}

\begin{cor}\label{cor:convmon_twistconv}
	The Grothendieck construction $\int_{\sf{Bun}_K(X)}\sDist$ restricts to a monoid $\int_{\Twist_K(X)} \sDist$ over $\Twist_K(X)$. More explicitly, this monoid is 
	\[
	\int_{\Twist_K(X)} \sDist\cong\coprod_{\eta\in\Twist_K(X)}\sDist(\pi_\eta)
	\]
	with multiplication $(\eta,p)\cdot (\xi,q)=(\eta\xi,p\ast q)$. 
\end{cor}

Similarly, if we restrict to the submonoid $\on{Twist}_K^0(X)$ of twisting functions which yield trivial bundles, Corollary \ref{cor:Theta_monoidal_functor} then yields

\begin{cor}
	The map 
	\[
	\begin{tikzcd}
		\Theta: &[-3em] {\displaystyle\coprod_{\eta \in \Twist_K^0(X)}D_R(\Sec(\pi_\eta))} \arrow[r] & {\displaystyle\coprod_{\eta \in \Twist_K(X)}\sDist(\pi_\eta)}
	\end{tikzcd}
	\]
	is a monoid homomorphism. 
\end{cor}

\subsection{Equivariant distributions}

We now provide a general construction of a morphism of bundle scenarios associated to a pair of bundles $\pi_E:E\to X$ and $\pi_F:F\to X$. 

Fix a simplicial Abelian group $K$. We will write the operation of $K$ multiplicatively to better accord with the standard notation for actions of $K$ on other simplicial sets.

\begin{defn}\label{defn:invprod_in_PB}
	Let $\pi_E:E\to X$ be a principal $K$-bundle, and let $a,b\in \pi_E^{-1}(x)$ be elements in the same fiber. We define $ba^{-1}\in K$ to be the unique element such that 
	\[
	ba^{-1}\cdot a=b. 
	\] 
\end{defn}

\begin{lem}
	Let $\pi_E:E\to X$ be a principal $K$-bundle, and let $a,b,c\in \pi_E^{-1}(x)$ be elements in the same fiber. Then the following computation rules hold. 
	\begin{enumerate}
		\item For $k\in K$, $(k\cdot b)a^{-1}=k(ba^{-1})$. 
		\item $(ba^{-1})(ac^{-1})=bc^{-1}$. 
		\item $(ba^{-1})^{-1}=ab^{-1}$. 
		\item $aa^{-1}=1$. 
	\end{enumerate}
\end{lem}

\begin{proof}
	If $h\cdot a=b$, then $(kh)\cdot a=kb$. By the freeness of the $K$-action, this implies $(k\cdot b)a^{-1}=k(ba^{-1})$, i.e., rule (1) holds. Similarly, if $kc=a$ and $ha=b$, then $hkc=b$, so that rule (2) holds. Rule (3) is a consequence of (2) and the uniqueness of inverses. Rule (4) holds since $1\cdot a=a$ for any $a\in E$. 
\end{proof}

Given two principal $K$-bundles $\pi_E:E\to X$ and $\pi_F:F\to X$, we can form the commutative diagram 
\[
\begin{tikzcd}
	E\arrow[d,"\pi_E"'] & E\times_XF\arrow[d,"\pi_{FE}"]\arrow[l,"\pi_{EF}"']\arrow[r] & E\otimes_KF \arrow[d,"\pi_{E\otimes F}"] \\
	X & F\arrow[l,"\pi_F"]\arrow[r,"\pi_F"'] & X 
\end{tikzcd}
\]
where the upper rightward
map is the quotient map. The 
right-hand
square induces a canonical map $E \times_X F\to (E\otimes_K F)\times_X F$ by the universal property of pullback. We thus obtain a morphism of bundle scenarios 
\[
\begin{tikzcd}
	E\arrow[d,"\pi_E"'] & E\times_XF\arrow[d,"\pi_{FE}"]\arrow[l,"\pi_{EF}"']\arrow[r,"\alpha"] & (E\otimes_KF)\times_X F \arrow[d,"\pi_2"] \\
	X & F\arrow[l,"\pi_F"]\arrow[r,"="'] & F.
\end{tikzcd}
\]
Explicitly, $\alpha$ sends $ (e,f)$ to $([e,f],f)$. 

This map of bundle scenarios induces a corresponding map on sets of simplicial distributions, which we will denote by
\[
\begin{tikzcd}
	\gamma_{E,F}: &[-3em] \sDist(\pi_E) \arrow[r] & \sDist(\pi_2). 
\end{tikzcd}
\]
Explicitly, given $p\in \sDist(\pi_E)$, we compute 
\[
\begin{aligned}
	\gamma_{E,F}(p)_f([e,j],\ell) &= \sum_{(u,v)\in \alpha^{-1}([e,j],\ell)} p_{\pi_F(f)}(u) \delta_{f,v} \\
	& =\sum_{u: [u,f]=[e,j]} p_{\pi_F(f)}(u)\delta_{f,\ell} \\
	& =\sum_{k:k^{-1}f=j} p_{\pi_F(f)}(k^{-1}e)\delta_{f,\ell}
\end{aligned}
\]
since $k$ acts freely transitively on $F$, there is a unique $k$ such that $kj=f$. Denoting this $k$ by $(fj^{-1})$ as in Definition \ref{defn:invprod_in_PB}, we thus obtain
\[
\gamma_{E,F}(p)_f([e,j],\ell)= p_{\pi_F(f)}((jf^{-1})\cdot e)\delta_{f,\ell}.
\]

We further note that the projection $D_R(\pi_1):D_R((E\otimes_KF)\times_X F)\to D_R(E\otimes_K F)$ induces a map 
\[
\begin{tikzcd}
	(\pi_1)_\ast: &[-3em] \sDist(\pi_2) \arrow[r] & \sDist(F,E\otimes_K F). 
\end{tikzcd}
\]

\begin{lem}\label{lem:Bundist_is_equivar_tensor_dist}
	The map $(\pi_1)_\ast\circ \gamma_{E,F}$ is injective. A simplicial distribution $p:F\to D(E\otimes_K F)$  is in the image of $(\pi_1)_\ast \circ \gamma_{E,F}$ if and only if  
\begin{enumerate}
\item[(1)] {$p$} is $K$-equivariant and 

\item[(2)] it makes the diagram 
	\[
	\begin{tikzcd}
		& D_R(E\otimes_K F)\arrow[d,"D_R(\pi_{E\otimes F})"] \\
		F \arrow[r,"\delta\circ \pi_F"']\arrow[ur,"p"] & D_R(X) 
	\end{tikzcd}
	\] 
	commute. 
\end{enumerate}		
\end{lem}

\begin{proof}
	We first show that the map $(\pi_1)_\ast \circ \gamma_{E,F}$ is injective. For brevity, let us write $\zeta=(\pi_1)_\ast \circ \gamma_{E,F}$. Suppose $p,q\in \sDist(\pi_E)$ such that $\zeta(p)=\zeta(q)$. Then, for any $f\in f$ and any $[e,j]\in E\otimes_K F$, we have 
	\[
	\begin{aligned}
	\zeta(p)_f([e,j]) &=p_{\pi_F(f)}((jf^{-1})\cdot e) \delta_{\pi_F(f),\pi_{E\otimes F}([e,j])} \\
	&=q_{\pi_F(f)}((jf^{-1})\cdot e) \delta_{\pi_F(f),\pi_{E\otimes F}([e,j])}\\
	&=\zeta(q)_f([e,j]).
\end{aligned}
	\]
	Given $x\in X$ and $g\in \pi_E^{-1}(x)$, we can choose any $f\in \pi_F^{-1}(x)$. Applying the formulae above to $f$ and $[g,f]$. we obtain 
	\[
	\begin{aligned}
		\zeta(p)_f([g,f])&=p_x((ff^{-1})\cdot g) \delta_{x,x}\\
		& =p_x(g)  
	\end{aligned}
	\]
	and similarly,
	\[
	\zeta(q)_f([g,f])=q_x(g). 
	\]
	Thus, $p=q$, and the map is injective, as desired. 
	
	We then check that given $p\in \sDist(\pi_E)$, $\zeta(p)$ has the two requisite properties. The term $\delta_{\pi_F(f),\pi_{E\otimes F}([e,j])}$ in the formula above shows that $\zeta(p)$ does indeed make the desired diagram commute. To see equivariance, we let $k\in K$ and compute
	\[
	\begin{aligned}
		\zeta(p)_{k\cdot f}([e,j])&= p_{\pi_F(k\cdot f)}((j(k\cdot f)^{-1})\cdot e) \delta_{\pi_F(k\cdot f),\pi_{E\otimes F}([e,j])}\\
		& =p_{\pi_F(f)}(k^{-1}\cdot((jf^{-1})\cdot e)) \delta_{\pi_F(f ),\pi_{E\otimes F}([e,j])}\\
		&= (\delta_k\ast \zeta(p)_f)([e,j]) 
	\end{aligned}
	\]
	so that the map is equivariant. 
	
	Finally, given $p:F\to D_R(E\otimes_KF)$ satisfying our two conditions, define a simplicial distribution $\overline{p}\in \sDist(\pi_E)$ defined by 
	\[
	\overline{p}_x(g)=\begin{cases}
		p_f([g,f]) & \pi_E(g)=x\\
		0 & \text{else},
	\end{cases}
	\]
	where $f\in F$ is any element such that $\pi_F(f)=x$. To see that this is well defined, note that, choosing another element $k\cdot f$, we have 
	\[
	p_{kf}([g,kf])= \delta_k\ast p_f([g,kf])=p_f(k^{-1} [g,kf])=p_f([g,k^{-1}kf])=p_f([g,f])
	\] 
	so that $\overline{p}_x(g)$ is well-defined. Similarly, we can compute 
	\[
	D_R(\pi_E)(\overline{p}_x)(y)= \sum_{e\in \pi_E^{-1}(y)}\overline{p}_x(e)=\delta_{x,y} \sum_{e\in \pi_E^{-1}(x)}p_f([e,f]) 
	\]
	where we can choose a fixed $f\in \pi_F^{-1}(x)$ for the last formula. Since there is a bijection between pairs $(e,f)$ with $e\in \pi_E^{-1}(x)$ for fixed $f$ and equivalence classes $[e,f]\in \pi_{E\otimes F}^{-1}(x)$, the latter sum can be written as 
	\[
	\delta_{x,y}\sum_{[e,j]\in \pi_{E\otimes F}^{-1}(x)} p_f([e,j])=\delta_{x,y}. 
	\] 
	Thus, the diagram 
	\[
	\begin{tikzcd}
		& D_R(E)\arrow[d,"D_R(\pi_E)"] \\
		X \arrow[r,"\delta"']\arrow[ur,"\overline{p}"] & D_R(X)
	\end{tikzcd}
	\]
	commutes, and $\overline{p}\in \sDist(\pi_E)$. 
	
	Finally, we compute 
	\[
	\begin{aligned}
		\zeta(\overline{p})_f([e,j])&=\overline{p}_{\pi_F(f)}((jf^{-1})\cdot e)\delta_{\pi_F(f),\pi_{E\otimes F}([e,j])}\\
		&=p_f([(jf^{-1})\cdot e,f])\\
		&=p_f([e,(j f^{-1})\cdot f])= p_f([e,j])
	\end{aligned}
	\]
	so that $\zeta(\overline{p})=p$. 
\end{proof}

\begin{defn}\label{defn:equivar_dist_notation}
	For $F,B$ simplicial sets equipped with $K$-actions, we denote by $\sDist_K(F,B)$ the set of $K$-equivariant simplicial distributions $F\to D_R(B)$, where $D_R(B)$ is equipped with the $K$ action given by 
	\[
	k\cdot p=\delta_k\ast p. 
	\] 
	
	For $\chi:F\to X$ and $\beta:B\to X$ two $K$-equivariant	
	simplicial maps, we denote by $\sDist_K^X(F,B)\subset \sDist_K(F,B)$ the subset of those equivariant distributions $p$ such that the diagram 
	\[
	\begin{tikzcd}
		& D_R(B) \arrow[d,"D_R(\beta)"] \\
		F \arrow[r,"\delta \circ \chi"'] \arrow[ur,"p"]& D_R(X)
	\end{tikzcd}
	\]
	commutes.
\end{defn}

\begin{thm}\label{thm:twisted as equivariant}
	For a $K$-valued twisting function $\eta$ on $X$, there is a bijection
	\[
	\sDist(\pi_\eta)\cong \sDist_K(K\times_{\eta^{-1}} X, K)
	\]
	given on elements by sending $p:X\to K\times_\eta X$ to the distribution $\widetilde{p}$ with
	\[
	\widetilde{p}_{k,x}(h)=p_x(hk^{-1},x).
	\]
\end{thm}

\begin{proof}
	The canonical isomorphism 
	\[
	\begin{tikzcd}[row sep=0em]
		(K\times_\eta X)\otimes_K (K\times_{\eta^{-1}} X) \arrow[r,"\cong"] & K\times X \\
		{[(h,x),(k,x)]}\arrow[r,mapsto] & (hk,x)
	\end{tikzcd}
	\]
	{induces} an isomorphism 
	\[
	\sDist_K^X(K\times_{\eta^{-1}} X,(K\times_\eta X)\otimes_K (K\times_{\eta^{-1}} X))\cong \sDist_K^X(K\times_{\eta^{-1}} X,K\times X).
	\] 
	Applying Lemma \ref{lem:Bundist_is_equivar_tensor_dist}, we obtain a bijection 
	\[
	\sDist(\pi_\eta)\cong \sDist_K^X(K\times_{\eta^{-1}} X,K\times X)
	\]
	which sends $p\in \sDist(\pi_\eta)$ to the distribution $\overline{p}$ defined by 
	\[
	\overline{p}_{(k,x)}(h,y)=\delta_{x,y} p_x(k^{-1}h,x).
	\]
	Since $K\times X$ is an untwisted product, there is a canonical isomorphism 
	\[
	\sDist_K^X(K\times_{\eta^{-1}} X,K\times X)\cong \sDist_K(K\times_{\eta^{-1}} X,K)
	\]
	completing the proof.
\end{proof}


\begin{ex}\label{ex:quantum}
Key examples of twisted distributions come from quantum theory. See for example \cite{watrous2018theory} for basic definitions. 
Let $\hH$ denote a finite-dimensional Hilbert space. We will write $\Proj(\hH)$ for the set of projectors.  
Analogous to the distribution functor we can define a functor  
$$
\begin{tikzcd}
		P_\hH: &[-3em] \catSet \arrow[r] & \catSet 
	\end{tikzcd}
$$
	which sends $X$ to the set of {\it projective measurements}
	\[
	P_\hH(X)=\left\lbrace \Pi:X\to \Proj(\hH)\;\mid \;\Pi \text{ {with finite support}, }\sum_{x\in X} \Pi(x)=\one\right\rbrace
	\]
	and sends $f: X\to Y$ to the map 
	\[
	\begin{tikzcd}[row sep=0em]
		P_\hH(f):&[-3em] P_\hH(X) \arrow[r] & P_\hH(Y) \\
		 & \Pi \arrow[r,mapsto] & \left(y \mapsto {\displaystyle\sum_{x\in f^{-1}(y)}} \Pi(x)\right).
	\end{tikzcd}
	\]
We can easily extend this functor to simplicial sets. Given a density operator $\rho$ (positive and trace $1$ linear operator), also called a quantum state, we have a natural transformation
$$ 
\begin{tikzcd}
		\rho_*: &[-3em] P_\hH \arrow[r] & D 
	\end{tikzcd}
$$ 	
whose components are given by $(\rho_*)_Y: P_\hH(Y)\to D(Y)$ that sends $\Pi \mapsto (y \mapsto \Tr(\rho\Pi(y)))$. For a simplicial group $K$ the corresponding map $P_\hH(K)\to D(K)$ is in fact $K$-equivariant. We can see this distribution as a twisted distribution. For a group $G$ the set $P_\hH(G)$ admits a  natural $G$ action: For $g\in G$ and $\Pi\in P_\hH(G)$ we have $g\cdot \Pi(h)=\Pi(g^{-1}\cdot h)$. Then we define 
$$
\bar P_\hH(G) = P_\hH(G)/\sim
$$
to be the quotient under this action. For simplicial sets the quotient gives a principal $K$-bundle
\begin{equation}\label{eq:PHK bundle}
\begin{tikzcd}
		\pi: &[-3em] P_\hH(H) \arrow[r] & \bar P_\hH(K) .
	\end{tikzcd}
\end{equation}
Thus the isomorphism in Theorem \ref{thm:twisted as equivariant} allows us to obtain a twisted distribution  
	\[
	\begin{tikzcd}
		& D(P_\hH(K)) \arrow[d,"D(\pi)"] \\
		\bar P_\hH(K) \arrow[r,"\delta"] \arrow[ur,"p_\rho"]& D(\bar P_\hH(K))
	\end{tikzcd}
	\]
An important special case is $K=N\ZZ_d$. Then there is an alternative description of $P_\hH(N\ZZ_d)$. For a group $G$, let us write $N(\ZZ_d,G)\subset NG$ for the simplicial subset consisting of $n$-simplices given by $n$-tuples $(g_1,\cdots,g_n)$ of pairwise commuting group elements such that $g_i^d=1$ for every $i$.	For the unitary group $U(\hH)$, the spectral decomposition of the commuting tuples of operators give an isomorphism of simplicial sets
$$
\begin{tikzcd}
		 &[-3em] N(\ZZ_d,U(\hH)) \arrow[r,"\cong"] & P_\hH(N\ZZ_d) .
	\end{tikzcd}
$$
The simplicial group $N\ZZ_d$ acts freely on $N(\ZZ_d,U(\hH))$ via the canonical inclusion induced by $1\mapsto e^{2\pi/d}\one$. Then the isomorphism above extends to an isomorphism of principal bundles
$$
\begin{tikzcd}
N(\ZZ_d,U(\hH)) \arrow[d]\arrow[r,"\cong"]  &  P_\hH(N\ZZ_d) \arrow[d] \\
\bar N(\ZZ_d,U(\hH)) \arrow[r,"\cong"] & \bar P_\hH(N\ZZ_d)
\end{tikzcd}
$$
where the left vertical map is the quotient just described.
\end{ex}

\section{Contextuality}
\label{sec:Contextuality}

In this section we introduce the notion of contextuality for simplicial distributions. We follow the framework of \cite{barbosa2023bundle}. To this end it is useful to introduce the interpretation of the simplicial sets involved. Given a simplicial set map $f:E\to X$, the base space represents the {\it measurement space}. The fiber $x^*(E)$ over a simplex $x:\Delta^n\to X$ corresponding to a measurement, represents the {\it space of outcomes}. Any physical experiment whose measurements and outcomes are discretely labeled can be described in this way; see \cite{okay2022simplicial} for more details.

\Def{\label{def:contextuality}
	A simplicial distribution $p$ on $f:E\to X$ is called \emph{non-contextual} if it lies in the image of
	$$
	\Theta: D_R (\Sec(f)) \to \sDist(f).
	$$
	Otherwise, it is called \emph{contextual}.  
}

When 
{$\eta\sim 0$}
the bundle ${K\times_\eta X}\to X$ is isomorphic to the projection map {$\pi_2:K\times X\to X$}. In this case simplicial distributions on $\pi_2$ can be identified with simplicial set maps $X\to D_R({K})$ recovering the original definition {of contextuality} in \cite{okay2022simplicial}.

The difficulty in studying contextuality for twisted distributions is that, unless $\eta\sim 0$, the set $\Sec(\pi_\eta)$ is empty. As a result, in most interesting cases, there are no {``obvious"} classical distributions. To address this, we examine distributions locally\footnote{{Here, ``local" is used in the topological sense as the opposite of ``global." Note that in the physics context, the word ``local" is also used for a special case of non-contextuality. To avoid confusion, we will use the word ``relative" in the topological context.}}. Throughout, fix a simplicial group $K$ and a $K$-principal bundle $\pi:E\to X$.

\begin{defn}\label{def:local contextuality}
	A morphism $\phi:U\to X$ of simplicial sets is called \emph{trivializing} (or \emph{$\pi$-trivializing}, when the bundle is not clear from context) when the pullback $E\times_X U\to U$ of $\pi:E\to X$ along $\phi$ is a trivial bundle. 
	
	If $\phi:U\to X$ is a $\pi$-trivializing morphism, we will call a simplicial distribution $p\in \sDist(\pi)$ \emph{$\phi$-relatively non-contextual} if $\psi^\phi_E(p)$ is non-contextual (here, $\psi^\phi_E$ is the map of Proposition \ref{prop:functoriality_bundle_scenarios}).  Note that since the restriction maps are convex, the set of $\phi$-relatively non-contextual distributions is itself a convex subset of $\sDist(\pi)$. We will call $p$ \emph{$\phi$-
{relatively}	
	deterministic} if $\psi^\phi_E(p)$ is deterministic.  
\end{defn}

\begin{ex}Basic examples of trivializing simplicial set maps that prove to be useful in applications are as follows.
	\begin{enumerate}
		\item Given a subset $U\subset X$ such that $U$ is weakly contractible, the morphism $U\to X$ is $\pi$-trivializing for every $K$-bundle over $X$.  
		\item A collection $\mathcal{U}:=\{U_i\}_{i=1}^k$ of subsets of $X$ is called a \emph{cover} of $X$ if the canonical map 
		\[
		\begin{tikzcd}
			\phi:&[-3em] {\displaystyle\coprod_{i=1}^k} U_i\arrow[r] & X 
		\end{tikzcd}
		\] 
		is surjective. If each inclusion $U\to X$ is $\pi$-trivializing, than so is $\phi$, in which case we say that $\mathcal{U}$ is a \emph{$\pi$-trivializing cover}. 
		\item By the definition of a principal $K$-bundle, the square 
		\[
		\begin{tikzcd}
			K\times X \arrow[r]\arrow[d,"\on{pr}_2"'] & E \arrow[d,"\pi"] \\
			E\arrow[r,"\pi"'] & X
		\end{tikzcd}
		\]
		is pullback. Thus, the morphism $\pi$ is always $\pi$-trivializing. 
	\end{enumerate}	
\end{ex}


\subsection{Twisting and collapsing}

In the case where $K=N(H)$ is the nerve of an Abelian group, one may relate distributions with specified deterministic behavior on a simplicial subset of the measurement space to twisted distributions on the corresponding quotient. To this end, let us fix an Abelian group $H$, a measurement space $X\in\catsSet$, a simplicial subset $Z\subset X$, and a normalized 2-cocycle $\alpha:X_2\to H$. For $\pi_\eta:K\times_\eta X\to X$, it will be convenient to write $\sDist_\eta(X)=\sDist(\pi_\eta)$ and $\Sec_\eta(X)=\Sec(\pi_\eta)$. To emphasize that $\eta$ comes from a cocycle $\alpha$ we will sometimes use the notation $\sDist_\alpha(X)$ and $\Sec_\alpha(X)$, respectively.

We will be working with the cofiber sequence 
\[
\begin{tikzcd}
	Z\arrow[r,"i"] & X \arrow[r,"j"] & X/Z 
\end{tikzcd}
\]
of simplicial sets throughout, and the corresponding short exact sequence of $H$-valued cochain complexes\footnote{Note that this sequence is \emph{not} actually exact in degree $0$, as the short exact sequence on cohomology only holds for relative cohomology. However, we will not at any point make use of $0$-cochains, and so it is sufficient for us to look at the truncated cochain complexes presented here. } 
\[
\begin{tikzcd}
	{C}_1(X/Z,H) \arrow[r,"j^\ast"]\arrow[d,"\partial"'] & {C}_1(X,H) \arrow[r,"i^\ast"]\arrow[d,"\partial"] & {C}_1(Z,H)\arrow[d,"\partial"]\\
	{C}_2(X/Z,H) \arrow[r,"j^\ast"]\arrow[d,"\partial"'] & {C}_2(X,H) \arrow[r,"i^\ast"]\arrow[d,"\partial"] & {C}_2(Z,H)\arrow[d,"\partial"] \\
	\vdots & \vdots & \vdots 
\end{tikzcd}
\]

First, let us assume that $i^\ast(\alpha)=0$. Then there is a unique 2-cocycle $\beta:(X/Z)_2\to H$ such that $j^\ast(\beta)=\alpha$.  Denote by $\eta$ the twisting function on $X/Z$ corresponding to $\beta$, and note that, defining $\eta|_{X}= \eta\circ j$ , we have that {$\eta|_X$} is the twisting function corresponding to $\alpha$. We thus obtain a commutative diagram 
\[
\begin{tikzcd}
	NH\times Z \arrow[r]\arrow[d,"\on{pr_2}"'] & 	NH\times_{\eta|_X} X \arrow[r]\arrow[d,"\pi_\alpha"'] & NH\times_{\eta} X/Z\arrow[d,"\pi_\eta"] \\
	Z \arrow[r,"i"'] & X \arrow[r,"j"'] & X/Z 
\end{tikzcd}
\]
in which each square is pullback, and the $N(H)$-components of the top morphisms are identities. By the functoriality of simplicial distributions on $\sf{bScen}_{N(H)}$, this provides a sequence of maps of convex sets of twisted distributions
\[
\begin{tikzcd}
	\sDist_\eta(X/Z) \arrow[r,"j^\ast"] & \sDist_{\eta|_X}(X) \arrow[r,"i^\ast"] & \sDist_0(Z). 
\end{tikzcd}
\] 
Since we can consider the delta distribution $\delta^{\varphi_0}$ on the zero section $\varphi_0$ as a basepoint of $\sDist_0(Z)$, it is possible to define the kernel of $i^\ast$, and thus, it is possible to ask whether this sequence is short exact. 

\begin{lem}
	The induced sequence 
	\[
	\begin{tikzcd}
		\sDist_\eta(X/Z) \arrow[r,"j^\ast"] & \sDist_{\eta|_X}(X) \arrow[r,"i^\ast"] & \sDist_0(Z)
	\end{tikzcd}
	\] 
	of convex sets is left-exact. 
\end{lem}

\begin{proof}
	We first note that $j\circ i:Z\to X/Z$ factors through the one-point space $\Delta^0$. Since there is a unique $0$-twisted distribution on $\Delta^0$ (the delta on the zero section), and the map $Z\to \Delta^0$ induces a pullback square sending this distribution to $\delta^{\varphi_0}$, we see that any $p\in \sDist_\eta(X/Z)$ satisfies 
	\[
	i^\ast(j^\ast(p))=\delta^{\varphi_0}. 
	\] 
	
	On the other hand, suppose that $p\in \sDist_{\eta|_X}(X)$ is a distribution such that $i^\ast(p)=\delta^{\varphi_0}$. This means the composite map 
	\[
	\begin{tikzcd}
		q:&[-3em] Z\arrow[r,"i"] & X \arrow[r,"p"] & D_R(NH\times_{\eta_X} X) 
	\end{tikzcd}
	\]
	satisfies 
	\[
	q_z(h,x) =\delta_{h,0}\delta_{z,x}
	\]
	for any $z\in Z$ and $h\in N(H)$. We then define a putative twisted distribution $\tilde{p}\in \sDist_\eta(X/Z)$ by setting 
	\[
	\tilde{p}_x(h,x)=p_{\tilde{x}}(h,\tilde{x})
	\]
	for some chosen $\tilde{x}$ such that $j(\tilde{x})=x$. Note that it does not matter which $\tilde{x}$ we choose, so long as the choices are consistent. It is immediate that 
	\[
	\sum_{h\in N(H)} \tilde{p}_x(h,x)=\sum_{h\in N(H)} {p}_{\tilde{x}}(h,\tilde{x})=1 
	\]
	so that this is, indeed a distribution. It remains for us to show it is simplicial. 
	
	By definition of the quotient, there is a unique 0-simplex $z\in X/Z$ such that $j\circ i$ factors through $z$. It is immediate from the definitions that $\tilde{p}_x(h,x)$ is compatible with degeneracy maps when $x$ is not degenerate on $z$, and compatible with all face maps $d_i$ such that $d_i(x)$ is not degenerate on $z$.  We thus have three cases to consider.
	\begin{itemize}
		\item First, let $\phi$ be a surjective map in $\Delta$. Then 
		\[
		\tilde{p}_{d_i(\phi^\ast(z))}(h,x)= \delta_{h,0}\delta_{\phi^\ast(z),x} 
		\]
		and similarly, 
		\[
		\begin{aligned}
			d_i(\tilde{p}_{\phi^\ast(z)})(d,x)& = \sum_{(k,y)\in d_i^{-1}(h,x)} \tilde{p}_{\phi^\ast(z)}(k,y)\\
			&= \sum_{(k,y)\in d_i^{-1}(h,x)} p_{\widetilde{\phi^\ast(z)}}(k,\tilde{y})\\
			& = \sum_{(k,y)\in d_i^{-1}(h,x), \tilde{y}\in Z} p_{\widetilde{\phi^\ast(z)}}(k,\tilde{y})+ \sum_{(k,y)\in d_i^{-1}(h,x), \tilde{y}\notin Z} p_{\widetilde{\phi^\ast(z)}}(k,\tilde{y})\\
			&= \sum_{(k,y)\in d_i^{-1}(h,x), \tilde{y}\in Z} p_{\widetilde{\phi^\ast(z)}}(k,\tilde{y})\\
			&= \sum_{(k,y)\in d_i^{-1}(h,x), \tilde{y}\in Z} \delta_{k,0} \delta_{\tilde{y},\widetilde{\phi^\ast(z)}} \\
			&= \delta_{h,0}\delta_{\phi^\ast(z),x}.
		\end{aligned}
		\] 
		\item A similar computation to the above holds for the $\phi^\ast(z)$ and a degeneracy map $s_i$. 
		\item The final case is that $x\notin Z$, but $d_i(x)=\phi^\ast(z)$ for a surjective map $\phi$ in $\Delta$. In this case, we can compute 
		\[
		\begin{aligned}
			d_i(\tilde{p}_{x}) (h,y)& = \sum_{(k,w)\in d_i^{-1}(h,y)} \tilde{p}_x(k,w)\\
			& = \sum_{(k,w)\in d_i^{-1}(h,y)} p_{\tilde{x}}(k,\tilde{w})\\
			& = \sum_{(k,u)\in d_i^{-1}(h,\tilde{y})} p_{\tilde{x}}(k,u)\\
			&= p_{d_i(\tilde{x})}(h,\tilde{y}) \\
			&= \tilde{p}_{d_i(x)}(h,y).
		\end{aligned}
		\]
	\end{itemize}
	
	By construction $j^\ast(\tilde{p})=p$, thus proving exactness at the middle term. It is immediate that $j^\ast$ is injective, completing the proof. 
\end{proof}

We denote by $\sDist_\eta(X,0)$ the (convex) set of $\eta$-twisted distributions on $X$ which satisfy $i^\ast(p)=\delta^{\varphi_0}$. 

\begin{thm}\label{thm:twist collapse iso for twistings}
	There is a convex isomorphism 
	\[
	\sDist_\eta(X/Z)\cong \sDist_\eta(X,0). 
	\]
\end{thm}

Note that this isomorphism respects contextuality. Let $\Sec_\eta(X,0)$ denote the set of sections $\varphi$ which satisfy $\varphi\circ i=\varphi_0$. We have a commutative diagram
\[
\begin{tikzcd}
D_R( \Sec_\eta(X,0) ) \arrow[r,"\Theta"] \arrow[d,"\cong"]  &  \sDist_\eta(X,0) \arrow[d,"\cong"] \\
D_R( \Sec_\eta(X/Z) ) \arrow[r,"\Theta"]  &  \sDist_\eta(X/Z)  
\end{tikzcd}
\]

We now modify our setup slightly, assuming not that $i^\ast(\alpha)=0$, but rather that $[i^\ast(\alpha)]=0$. We choose a normalized 1-cochain $\nu:Z_1\to H$ such that $\partial \nu=\alpha$, and define a normalized 1-cochain on $X$ by  
\[
\tilde{\nu}(x)=\begin{cases}
	\nu(x) & x\in Z_1 \\
	0 & \text{else}.
\end{cases}
\]
Then $i^\ast(\alpha-\partial\tilde{\nu})=0$, so that we can again find a normalized 2-cocycle $\beta:(X/Z)_2\to H$ such that $j^\ast(\beta)=\alpha-\partial\tilde{\nu}$.

We can then apply Theorem \ref{thm:twist collapse iso for twistings}
to $\alpha-\partial\tilde{\nu}$, obtaining a convex isomorphism 
\begin{equation}\label{eq:previous work}
\sDist_\beta(X/Z)\cong \sDist_{\alpha-\partial\tilde{\nu}}(X,0), 
\end{equation}
where we emphasize the $2$-cocycle in the notation.
We will use the monoid homomorphism of Corollary \ref{cor:Theta_monoidal_functor}
\[
\begin{tikzcd}
	\Theta: &[-3em] {\displaystyle\coprod_{\eta\in\Twist_K(X)}} D_R(\Sec(\pi_\eta)) \arrow[r] & {\displaystyle\coprod_{\eta\in\Twist_K(X)}} \sDist_\eta(X)   .
\end{tikzcd}
\]
In particular, the delta distributions on sections of trivial bundles are invertible in the convex monoid of twisted distributions. Specializing to twistings coming from cocycles and denoting by $\sDist_\alpha(X,\nu)$ the $\alpha$-twisted distributions on $X$ which restrict to $\delta^{\varphi_\nu}$ on $Z$, we then obtain the following.

\begin{cor}\label{cor:twist and collapse cocycle}
	There is a convex isomorphism  
	\[
	\sDist_\beta(X/Z)\cong \sDist_\alpha(X,\nu). 
	\]
\end{cor}

\begin{proof}
	By isomorphism (\ref{eq:previous work}),
	it is sufficient to show that there is a convex isomorphism
	\[
	\sDist_\alpha(X,\nu)\cong \sDist_{\alpha-\partial\widetilde{\nu}}(X,0).
	\]
	However, $\delta^{\varphi_{-\partial\widetilde{\nu}}}$ is an invertible element of the monoid of twisted distributions, with inverse $\delta^{\varphi_{\partial\widetilde{\nu}}}$. Multiplying $p\in \sDist_\alpha(X,\nu)$ by $\delta^{\varphi_{-\partial\widetilde{\nu}}}$ yields 
	\[
	p\cdot \delta^{\varphi_{-\partial\widetilde{\nu}}}\in \sDist_{\alpha-\partial\widetilde{\nu}}(X,0)
	\]
	and since $\delta^{\varphi_{-\partial\widetilde{\nu}}}$ is invertible, this map is an isomorphism, completing the proof.
\end{proof}

The isomorphism in Corollary \ref{cor:twist and collapse cocycle} restricts to a bijection of deterministic distributions. Writing $\Sec_\alpha(X,\nu)$ for sections that restrict to $\varphi_\nu$ on $Z$, we have a commutative diagram
$$
\begin{tikzcd}
D_R( \Sec_\alpha(X,\nu) ) \arrow[r,"\Theta"] \arrow[d,"\cong"]  &  \sDist_\alpha(X,\nu) \arrow[d,"\cong"] \\
D_R( \Sec_\beta(X/Z) ) \arrow[r,"\Theta"]  &  \sDist_\beta(X/Z)  
\end{tikzcd}
$$ 
This diagram can be used to analyze contextuality by collapsing a subspace at the cost of twisting the distributions. Some key examples in the case of $X$ is a $2$-dimensional simplicial complex are worked out in \cite{okay2023rank}.

\subsection{The Mermin {polytope}}
\label{sec:mermin polytope}

In this section we study a small example in full detail. This example is the paradigmatic example of contextual quantum distributions first introduced by Mermin \cite{mermin1993hidden}. Later in \cite{Coho} this construction is turned into a topological one, and finally in \cite{okay2022mermin} the distributions on this scenario are introduced and their contextual behavior are studied.

We let $M$ denote the measurement space of the Mermin scenario, that is, the simplicial set 
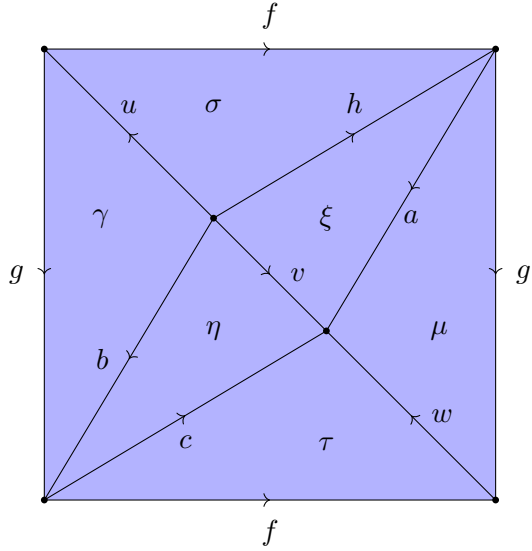
\begin{center}
	\begin{tikzpicture}[scale=1.5, decoration={
			markings,
			mark=at position 0.5 with {\arrow{>}}}]
		\path[fill=blue,opacity=0.3] (0,0) rectangle (4,4);
		\draw[postaction={decorate}] (0,0) to node[label=below:$f$]{} (4,0);
		\draw[postaction={decorate}] (4,4) to node[label=right:$g$]{} (4,0);
		\draw[postaction={decorate}] (0,4) to node[label=above:$f$]{} (4,4);
		\draw[postaction={decorate}] (0,4) to node[label=left:$g$]{} (0,0);
		\draw[postaction={decorate}]  (1.5,2.5) to node[label=above:$u$]{} (0,4);
		\draw[postaction={decorate}] (1.5,2.5) to node[label=left:$b$]{} (0,0);
		\draw[postaction={decorate}] (0,0) to node[label=below:$c$]{} (2.5,1.5);
		\draw[postaction={decorate}] (4,0) to node[label=right:$w$]{} (2.5,1.5);
		\draw[postaction={decorate}] (4,4) to node[label=below:$a$]{} (2.5,1.5);
		\draw[postaction={decorate}]  (1.5,2.5) to node[label=above:$h$]{} (4,4);
		\draw[postaction={decorate}]  (1.5,2.5) to node[label=right:$v$]{} (2.5,1.5);
		\foreach \x/\y in {0/0,4/0,4/4,0/4,1.5/2.5,2.5/1.5}{
			\path[fill=black] (\x,\y) circle (0.03);
		};
		\path (1.5,3.5) node {$\sigma$};
		\path (1.5,1.5) node {$\eta$};
		\path (0.5,2.5) node {$\gamma$};
		\path (2.5,0.5) node {$\tau$};
		\path (3.5,1.5) node {$\mu$};
		\path (2.5,2.5) node {$\xi$};
	\end{tikzpicture}
\end{center}
whose realization is a torus. Let $\beta:M_2\to \ZZ_2$ be the 2-cocycle which assigns $1$ to $\sigma$, $\eta$, and $\mu$, and $0$ otherwise. 
Let $\pi_\beta:E_\beta \to M$ be the bundle represented by $\beta$. 

\begin{rem}
Consider the $2\times 2$ Pauli matrices
$$
X =\left( \begin{matrix}
0 & 1 \\
1 & 0
\end{matrix} \right)\;\;\;
Y =\left( \begin{matrix}
0 & -i \\
i & 0
\end{matrix} \right)\;\;\;
Z =\left( \begin{matrix}
1 & 0 \\
0 & -1
\end{matrix} \right).
$$
We can define a simplicial set map
$$
 \begin{tikzcd}
	  \phi:&[-3em] M  \arrow[r] & \bar N(\ZZ_2,U((\CC^2)^{\otimes 2}))  
\end{tikzcd}
$$
by sending the $1$-simplices to the equivalence class (under $\pm\one$) of the following matrices 
$$
f\mapsto Z\otimes X,\;\;\; u\mapsto X\otimes Y,\;\;\; g\mapsto Y\otimes X,\;\;\; w\mapsto X\otimes Z.
$$
Then $\pi_\beta$ is precisely the pull-back of $N(\ZZ_2,U((\CC^2)^{\otimes 2}))\to \bar N(\ZZ_2,U((\CC^2)^{\otimes 2}))$ along $\phi$.
\end{rem}

We borrow a definition from \cite{okay2022mermin}, presented using the language of simplicial sets.

\begin{defn}
	A simplicial subset $U\subset M$ is called \emph{closed} if, whenever 
	$U$ contains two edges of a 2-simplex $\sigma$, then $U$ contains the 2-simplex $\sigma$. 	A simplicial subset $U\subset M$ is called \emph{non-contextual} if there is a 1-cochain $s:U_1\to \ZZ_{2}$ such that $\partial(s)=\beta|_{U}$.
\end{defn} 



\begin{lem}
	A simplicial subset $U\subset M$ is non-contextual if and only if it is $\pi_\beta$-trivializing.  
\end{lem} 

\begin{proof}
	This is simply the bijection between second cohomology classes and isomorphism classes of bundles. That is, a non-contextual subset $U\subset N$ is precisely one on which the cocycle $\beta|_U$ vanishes in cohomology. This is equivalent to the restricted bundle $U\times_M E_\beta$ being represented by the zero class in cohomology, and thus being trivial. 
\end{proof}



Consider the simplicial set $\widetilde{M}$ whose $n$-simplices consist of an $n$-simplex $\sigma$ in $M$ and a map 
\[
f:\on{sk}_1(\sigma)\to N(\ZZ_{2}) 
\]
where $\on{sk}_1(\sigma)$ denotes the 1-skeleton of $\sigma$.

\begin{lem}
	The simplicial set $\widetilde{M}$ is relatively 1-coskeletal over $M$. 
\end{lem}

\begin{proof}
	By construction, a simplex in $\widetilde{M}$ is uniquely determined by its 1-skeleton and the underlying simplex in $M$. The 1-coskeletal lifting problem follows immediately. 
\end{proof}

We define a map on 2-truncations 
\[
\begin{tikzcd}
	\on{tr}_2E_\beta\arrow[dr,"\pi_\beta"'] \arrow[rr,"\xi"] && \on{tr}_2\widetilde{M}\arrow[dl]\\
	& \on{tr}_2(M) & 
\end{tikzcd}
\] 
as follows. 

On both sides, the $0$-simplices are precisely those of $M$, and so we define the map on $0$-simplices to be the identity. On $1$-simplices, we send a pair $(\sigma,g)$ to the pair $(\sigma,g)$. On 2-simplicies, given $(\sigma,(g_{0,1},g_{1,2}))$, we send it to the simplex $\sigma$ with the map on the $1$-skeleton given by
\[
\begin{tikzcd}
	& \ast\arrow[dr,"g_{1,2}+\beta(\sigma)"] & \\
	\ast\arrow[ur,"g_{0,1}"]\arrow[rr,"g_{0,1}+g_{1,2}"'] & & \ast.
\end{tikzcd}
\]

\begin{lem}
	The map $\xi$ is a morphism of $2$-truncated simplicial sets. 
\end{lem}

\begin{proof}
	We check the simplicial identities. For the face and degeneracy maps between $0$-simplices and $1$-simplices, it is immediate that $\xi$ commutes with the maps. 
	
	Let $(\sigma,\{g_{0,1},g_{1,2}\})$ be a $2$-simplex of $E_\beta$. Then by definition  
	\[
	d_0(\xi(\sigma,\{g_{0,1},g_{1,2}\}))=(d_0(\sigma),g_{1,2}+\beta(\sigma))=\xi(d_0(\sigma),g_{1,2}+\beta(\sigma))=\xi(d_0(\sigma,\{g_{0,1},g_{1,2}\})).
	\] 
	The relations with $d_1$, $d_2$, $s_0$, and $s_1$ are immediate (since $\beta$ is normalized, the {$\beta(\sigma)$} in the definition of the map $\xi$ vanishes for degenerate $\sigma$).
\end{proof}

\begin{lem}
	The map $\xi$ is injective, and its image consists of all $0$- and $1$-simplices, together with those 2-simplices $(\sigma, g_{0,1},g_{1,2},g_{0,2})$ satisfying 
	\[
	g_{0,1}+g_{1,2}+g_{0,2}=\beta(\sigma).
	\]
\end{lem} 

\begin{proof}
	Injectivity is immediate from the definitions. Since every element of $\ZZ_{2}$ is 2-torsion, it is similarly immediate that the 2-simplices in the image of $\xi$ satisfy the given condition. 
	
	On the other hand, let $(\sigma, g_{0,1},g_{1,2},g_{0,2})$ be a 2-simplex of $\widetilde{M}$ satisfying 
	\[
	g_{0,1}+g_{1,2}+g_{0,2}=\beta(\sigma).
	\]
	Then $(\sigma,\{g_{0,1},g_{1,2}-\beta(\sigma)\})$ is a 2-simplex of $E_\beta$ which maps to the specified 2-simplex of $\widetilde{M}$. 
\end{proof}

We now denote by $\overline{M}$ the unique relative 2-coskeletal extension over $M$ of the image of $\xi$. 

\begin{cor}
	The map $\xi$ induces an isomorphism over $M$ between $E_\beta$ and $\overline{M}$. 
\end{cor}

The original construction in \cite{mermin1993hidden} uses a different representation. Let us write $\mM$ for the set of non-degenerate $1$-simplices and $\cC$ for the set of non-degenerate $2$-simplices of $M$. Here the elements of the first set represent measurements and those of the second set are referred to as {\it contexts}. The same information contained in the torus picture above can be depicted as a hypergraph:
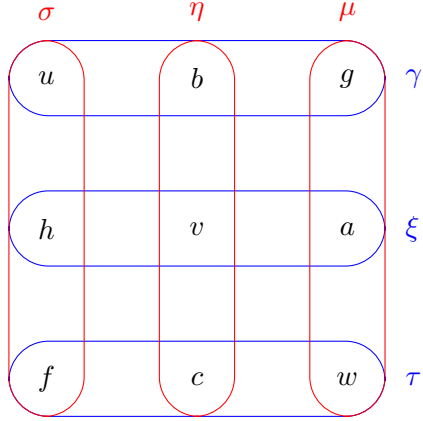
\begin{center}
	\begin{tikzpicture}
		\path  (0,0) node {$f$}; 
		\path  (2,0) node {$c$};
		\path  (4,0) node {$w$};  
		
		\path  (0,2) node {$h$}; 
		\path  (2,2) node {$v$};
		\path  (4,2) node {$a$}; 
		
		\path  (0,4) node {$u$}; 
		\path  (2,4) node {$b$};
		\path  (4,4) node {$g$}; 
		
		\draw[blue,rounded corners=15pt] (-0.5,3.5)  rectangle (4.5,4.5);   
		\draw[blue,rounded corners=15pt] (-0.5,1.5)  rectangle (4.5,2.5); 
		\draw[blue,rounded corners=15pt] (-0.5,-0.5)  rectangle (4.5,0.5);
		\path[blue] (4.5,0) node[label=right:$\tau$] {};
		\path[blue] (4.5,2) node[label=right:$\xi$] {};
		\path[blue] (4.5,4) node[label=right:$\gamma$] {};
		
		\draw[red,rounded corners=15pt] (3.5,-0.5)  rectangle (4.5,4.5);   
		\draw[red,rounded corners=15pt] (1.5,-0.5)  rectangle (2.5,4.5); 
		\draw[red,rounded corners=15pt] (-0.5,-0.5)  rectangle (0.5,4.5);
		\path[red] (0,4.5) node[label=above:$\sigma$] {};
		\path[red] (2,4.5) node[label=above:$\eta$] {};
		\path[red] (4,4.5) node[label=above:$\mu$] {};
	\end{tikzpicture}
\end{center} 

\noindent{}
Here, the red circles correspond precisely to those triangles on which the cocycle $\beta$ does not vanish.

\begin{defn}
	For $C\in \mathcal{C}$ a context in the classical Mermin square scenario $(\mathcal{M},\mathcal{C})$, we define
	\[
	O_\beta(C):=\left\lbrace s:C\to \ZZ_{2}\; \middle| \;\sum_{m\in C}s(m)=\beta(C) \right\rbrace\subset {\ZZ_{2}^C}
	\] 
	where we write $\ZZ_{2}^C$ for the set of $C$-indexed tuples in $\ZZ_2$.	
	Define the \emph{non-signaling polytope} to be 
	\[
	\on{NS}:=\left\lbrace \{p_C:\ZZ_{2}^C \to [0,1]\}_{C\in \mathcal{C}} \; \middle| \; \sum_{m\in C} p(m)=1, \; p_{C}|_{C\cap C^\prime}=p_{C^\prime}|_{C\cap C^\prime}  \right\rbrace .
	\]
	Define the $\beta$-nonsignaling polytope to be the subset $\on{NS}_\beta\subset \on{NS}$, on the $\{p_C\}_{C\in \mathcal{C}}$ such that 
	\[
	\on{supp}(p_C)\subset O_\beta(C)
	\]
	for all $C\in \mathcal{C}$. 
\end{defn}

\begin{lem}
	There is an isomorphism 
	\[
	\sDist(\overline{M}\to M)\cong \on{NS}_\beta. 
	\]	
\end{lem}

\begin{proof}
	A distribution $p_C$ supported on $O_\beta$ is precisely a distribution on the set of 2-simplices of $\overline{M}$ lying over the simplex of $M$ corresponding to $C$. The compatibility relations in $\on{NS}$ are then precisely those required to make the collection a simplicial set map 
	\[
	\on{tr}_2(M)\to \on{tr}_2(D(\overline{M}))
	\]
	over $\on{tr}_2(D(M))$. Since $M$ is 2-skeletal, this is equivalent to a map 
	\[
	M\to D(\overline{M})
	\]
	over $D(M)$, i.e., a simplicial distribution on $\overline{M}\to M$. 
\end{proof}

\begin{cor}
	There is an isomorphism 
	\[
	\sDist(E_\beta)\cong \on{NS}_\beta. 
	\]
\end{cor}

The vertex enumeration problem for the polytope $\sDist(E_\beta)$ is solved in \cite{okay2022mermin}. Up to combinatorial automorphisms of the polytope there are two kinds of vertices. Both of them are 
{relatively}
deteministic in the sense of Definition \ref{def:local contextuality}. The trivializing morphisms are given as follows:
\begin{itemize}
\item Let $U$ denote the $1$-dimensional simplicial subset of $M$ consisting of the $1$-simplices $u,c,a$. The inclusion $\phi:U\to M$ is 
trivializing. The vertex corresponding to this type restricts to a deterministic distribution on $U$.

\item Let $V$ denote the $2$-dimensional simplicial set consisting of the two non-degenerate simplices $\eta$ and $\xi$. Then $\phi:V\to X$ is trivializing and restriction of the vertex to $V$ is deterministic.
\end{itemize}
These trivializing morphisms uniquely determine the vertices. Moreover, in \cite{okay2023rank} a $2$-dimensional version of Corollary \ref{cor:twist and collapse cocycle} is used to solve the vertex enumeration problem.


\addcontentsline{toc}{section}{References}

\bibliographystyle{alpha}
\bibliography{bib.bib}

\appendix

\section{Monoidal categories and functors}

In this section, we will show that various functors and categories used in the main body are (lax) (symmetric) monoidal. 

\begin{pro}\label{prop:D_lax_mon_set_set}
	The functor $D:\catSet\to \catSet$ is lax symmetric monoidal with respect to the cartesian monoidal structure. The structure maps are 
	\[
	\begin{tikzcd}[row sep=1em]
		D(X)\times D(Y) \arrow[r,"m_{X,Y}"] & D(X\times Y) \\
		\ast \arrow[r,"\varepsilon"] & D(\ast) 
	\end{tikzcd}
	\]
	where $\ast$ denotes the singleton, and $\epsilon$ is the unique isomorphism between singletons. 
\end{pro}

\begin{proof}
	We first note that $m$ is natural. Given maps of sets $f:X\to Z$ and $g: Y\to W$ and $(p,q)\in D(X)\times D(Y)$ we have 
	\[
	\begin{aligned}
		m(D(f)\times D(g)(p,q))(z,w)&= D(f)(p)(z)\cdot D(g)(q)(w)\\
		&= \left(\sum_{x\in f^{-1}(z)} p(x)\right)\cdot \left(\sum_{y\in g^{-1}(w)} q(y)\right)\\
		&=\sum_{\substack{x\in f^{-1}(z)\\ y\in g^{-1}(w)}} p(x)q(y)\\
		&= \sum_{(x,y)\in(f\times g)^{-1}(z,w)}p(x)q(y)\\
		&= D(f\times g)(m(p,q))(z,w). 
	\end{aligned}
	\] 
	We then check associativity. Given $(p,q,r)\in D(X)\times D(Y)\times D(Z)$, we have 
	\[
	m((m\times \on{id}_{D(Z)})(p,(q,r)))((x,y),z)=(p(x)q(y))r(z)
	\]
	and 
	\[
	m((\on{id}_{D(X)}\times m)((p,q),r))(x,(y,z))= p(x)(q(y)r(z))
	\]
	and so the associativity of multiplication completes the argument. Unitality is immediate, and symmetry follows from the symmetry of multiplication on $[0,1]$.
\end{proof}

The following, though inessential for the present paper, is an immediate corollary, closely related to the work in \cite{HOSconvex}. For this reason, we make free use of the notion of biconvexity and the concomitant tensor product of convex sets developed in op. cit. 

\begin{cor}
	The lax symmetric monoidal structure on the functor $D:\catSet\to \catSet$ induces a lax symmetric monoidal structure on the functor $D:\catSet\to (\sf{CSet},\otimes)$. 
\end{cor}

\begin{proof}
	Since the maps $m_{X,Y}$ are all biconvex, this follows immediately from the universal property of $\otimes$. 
\end{proof}

\begin{thm}\label{thm:Bun_SMC}
	For a simplicial Abelian group $K$, the product $-\otimes_K-$ equips $\sf{Bun}_K(X)$ with the structure of a symmetric monoidal groupoid. The structure maps are given by unitors
	\[
	\begin{tikzcd}[row sep=0em]
		(K\times X)\otimes_K E \arrow[r, "\lambda_E"] & E \\
		{	[(k,x),e]} \arrow[r, mapsto] & k\cdot e
	\end{tikzcd}
	\]
	and 
	\[
	\begin{tikzcd}[row sep=0em]
		E \otimes_K (K\times X) \arrow[r,"\rho_E"] & E \\
		{[e,(k,x)]} \arrow[r,mapsto]& k\cdot e;
	\end{tikzcd}
	\]
	associators
	\[
	\begin{tikzcd}[row sep=0em]
		E\otimes_K(F\otimes_K G) \arrow[r,"\alpha_{E,F,G}"] & (E\otimes_K F)\otimes_K G \\
		{[e,[f,g]]}\arrow[r,mapsto] & {[[e,f],g]};
	\end{tikzcd}
	\]
	and the braiding 
	\[
	\begin{tikzcd}[row sep=0em]
		E\otimes_K F \arrow[r,"\beta_{E,F}"] & F\otimes_K E \\
		{[e,f]} \arrow[r,mapsto] & {[f,e]}.
	\end{tikzcd}
	\]
\end{thm}

\begin{proof}
	Since the associators and braiding are induced by those of the Cartesian monoidal structure on $\catSet$, the pentagon and hexagon identities will follow from the well-definedness of these associators and braiding. To show well-definedness, suppose $[e_1,[f_1,g_1]]=[e_2,[f_2,g_2]]$. Then there is $k\in K$ such that $ke_1=e_2$ and $k^{-1}\cdot[f_1,g_1]=[f_2,g_2]$. The latter relation implies that there is $h\in K$ such that $hk^{-1}f_1=f_2$ and $h^{-1}g_1=g_2$. Setting $\ell= h$ and $a=kh^{-1}$, we thus see that 
	\[
	a\ell e_1=e_2,\quad a^{-1}f_1=f_2,\quad \ell^{-1}g_1=g_2
	\]
	so that $[[e_1,f_1],g_1]=[[e_2,f_2],g_2]$ and the associator is well-defined. Similarly, for the braiding, if $[e_1,f_1]=[e_2,f_2]$, then there is $k\in K$ such that $ke_1=e_2$ and $k^{-1}f_1=f_2$. Setting $h=k^{-1}$ shows that $[f_1,e_1]=[f_2,e_2]$, and so the braiding is well-defined.
	
	It thus remains only for us to check the triangle identity. We show that the diagram 
	\[
	\begin{tikzcd}
		(E\otimes_K (K\times X))\otimes_K F\arrow[rr,"\alpha_{E,K\times X,F}"]\arrow[dr,"{\rho_E\otimes\on{id}}"'] & & E\otimes_K ((K\times X)\otimes_K F)\arrow[dl,"{\on{id}\otimes\lambda_F}"]\\
		& E\otimes_K F &
	\end{tikzcd}
	\]
	commutes. Starting with $[[e,(k,x)],f]$, applying the associator and $\on{id}\times \lambda_F$ yields $[e,k\cdot f]$. Applying, on the other hand, $\rho_E$, yields $[k\cdot e,f]$. By the definition of our equivalence relation $[k\cdot e,f]=[e,k\cdot f]$, and so the diagram commutes. 
\end{proof}

\begin{pro}\label{prop:Bun_sym_mon_funct_Set}
	The functor 
	\[
	\begin{tikzcd}
		\sDist: &[-3em] \sf{Bun}_K(X) \arrow[r] & \catSet
	\end{tikzcd}
	\]
	is lax symmetric monoidal with respect to the product $\otimes_K$ on $\sf{Bun}_K(X)$ and the cartesian monoidal structure on $\catSet$. The structure maps are 
	\[
	\begin{tikzcd}
		\mu_{E,F}: &[-3em] \sDist(\pi_E)\times \sDist(\pi_F) \arrow[r] & \sDist(\pi_{E\otimes_K F}) 
	\end{tikzcd}
	\]
	and the map
	\[
	\begin{tikzcd}
		\varepsilon : &[-3em] \ast \arrow[r] & \sDist(K\times X\to X)
	\end{tikzcd}
	\]
	which sends the unique element to the delta distribution $\delta^{\varphi_0}$ on the zero section. 
\end{pro}

\begin{proof}
	We check only the left unitality diagram. The right unitality diagram is similar. We compute the composite 
	\[
	\begin{tikzcd}
		\ast\times \sDist(\pi_E)\arrow[r,"\epsilon\times \on{id}"] & \sDist(X\times K)\times \sDist(\pi_E) \arrow[r,"\mu"] & \sDist((X\times K)\otimes_K E) \arrow[r,"\lambda_E"] & \sDist(\pi_E)
	\end{tikzcd}
	\]
	acting on $p\in \sDist(\pi_E)$. This is, effectively, the product of $(\delta^{\varphi_0},p)$. Evaluated on $e\in E$, we compute
	\[
	\begin{aligned}
		\lambda_E(\mu(\delta^{\varphi_0},p))_x(e)& = \sum_{k+y=e}\delta^{\varphi_0}_x(k,x)p_x(y) \\
		&= \sum_{0+y=e} p_x(y)=p_x(e)
	\end{aligned}
	\]
	so that the composite is the (Cartesian) unitor on $\sDist(\pi_E)$ as desired.
	
	Since $\mu$ is directly induced by the multiplication $m$ and the quotient map, the associativity follows easily from the associativity proven in Proposition \ref{prop:D_lax_mon_set_set}. Symmetry follows directly from the symmetry of multiplication on $[0,1]$.  
\end{proof}

\begin{pro}\label{pro:Cl_lax_mon}
	The maps $\ell$ and $u$ equip the functor
	\[
	\begin{tikzcd}[row sep=0em]
		\on{Cl}:&[-3em]\sf{Bun}_K^0(X) \arrow[r] & \catSet
	\end{tikzcd}
	\]
	with the structure of a lax symmetric monoidal functor.
\end{pro}

\begin{proof}
	It is sufficient to show that $\overline{\ell}$ and the zero section equip $\on{Sec}$ with a lax monoidal structure, since $\on{Cl}$ will then simply be $D_R\circ \on{Sec}$.
	
	For associativity, it is thus sufficient to check associativity for $\overline{\ell}$. Given sections $(\varphi,\psi,\zeta)\in \Sec(\pi_E)\times \Sec(\pi_F)\times \Sec(\pi_G)$, we see that both 
	\[
	\overline{\ell}((\overline{\ell}\times\on{id}_{\Sec(\pi_G)})(\varphi,\psi,\zeta))
	\] 
	and  
	\[
	\overline{\ell}((\on{id}_{\Sec(\pi_E)}\times\overline{\ell})(\varphi,\psi,\zeta))
	\] 
	can be equivalently identified with the composite of 
	\[
	\begin{tikzcd}
		(\varphi,\psi,\zeta): &[-3em] X \to E\times_X F\times_X G
	\end{tikzcd}
	\]
	with the quotient map to $E\otimes_KF\otimes_KG$. Symmetry follows similarly. 
	
	For unitality, we only write out the check of the left unitality diagram. The composite 
	\[
	\begin{tikzcd}
		\ast\times \Sec(\pi_E)\arrow[r,"\overline{u}\times \on{id}"] & \Sec(X\times K)\times \Sec(\pi_E) \arrow[r,"\overline{\ell}"] & \Sec((X\times K)\otimes_K E) \arrow[r,"\lambda_E"] & \Sec(\pi_E)
	\end{tikzcd}
	\]
	We see that a section $\varphi$ is sent to the section $\psi$ defined by 
	\[
	\psi(x)=\lambda_E([1,\varphi(x)])=\varphi(x)
	\]
	precisely as desired.
\end{proof}

\begin{pro}\label{prop:Twist_sym_mon_funct_bun}
	Viewing the commutative monoid $\on{Twist}_K(X)$  of $K$-valued twisting functions on $X$ as a discrete symmetric monoidal category, the functor  
	\[
	\begin{tikzcd}[row sep=0em]
		\Twist_K(X) \arrow[r] & \sf{Bun}_K(X) \\
		\eta \arrow[r,mapsto] & K\times_\eta X
	\end{tikzcd}
	\]
	is symmetric monoidal with structure maps given by the canonical isomorphism 
	\[
	K\times_0 X \cong K\times X 
	\]
	and the isomorphisms 
	\[
	(K\times_{\eta} X)\otimes_K(K\times_\xi X)\cong K\times_{\eta+\xi}X 
	\]
	constructed in the proof of Proposition \ref{pro:Addition_of_twists_tensor}.
\end{pro}

\section{Colimits and finite simplicial sets}

\begin{lem}[Eilenberg-Zilber Lemma]\label{lem:E-Z}
	Let $X$ be a simplicial set, and let $\sigma\in X_n$ be a simplex. Then there is a unique surjective map $\phi:[n]\to [k]$ in $\Delta$ and a unique non-degenerate simplex $\tau\in X_k$ such that $\sigma= \phi^\ast(\tau)$.
\end{lem} 

\begin{proof}
	There are many proofs in the literature. A textbook reference is \cite[Thm. 4.2.3]{FritschPiccinini}.
\end{proof}

\begin{lem}\label{lem:fin_sset_cat_simps_cofinal} 
	Let $X$ be a simplicial set of dimension $n$. Denote by $(\Delta_{/X})^{\leq n}$ the full subcategory of the category of simplices on the simplices of dimension less than or equal to $n$. Then the inclusion 
	\[
	\begin{tikzcd}
		(\Delta_{/X})^{\leq n}\arrow[r] & \Delta_{/X}	
	\end{tikzcd}
	\]
	is cofinal. 
\end{lem} 

\begin{proof}
	Let $\sigma\in X_k$ be a $k$-simplex of $X$. It will suffice for us to show that $(\Delta_{/X})^{\leq n}\times_{\Delta_{/X}}(\Delta_{/X})_{\sigma/}$ is non-empty and connected. If $k\leq n$, then this is immediate, since $(\Delta_{/X})^{\leq n}\times_{\Delta_{/X}}(\Delta_{/X})_{\sigma/}$ is then canonically isomorphic to $\left((\Delta_{/X})^{\leq n}\right)_{\sigma/}$. 
	
	In the case where $k\geq n$, let $(\phi,\tau)$ be the unique decomposition of $\sigma$ guaranteed by the Eilenber-Zilber Lemma \ref{lem:E-Z}. Then $\phi:\sigma \to \tau$ is an element of $(\Delta_{/X})^{\leq n}\times_{\Delta_{/X}}(\Delta_{/X})_{\sigma/}$. Now suppose there is another element $\psi:\sigma\to \rho$.  
	
	We can factor $\psi=\psi_2\circ \psi_1$, where $\psi_2$ is injective, and $\psi_1$ is surjective. Then $\psi_1^\ast(\psi_2^\ast(\rho))=\sigma$. However, since $\psi_2$ is injective, this means that the dimension of $\psi_2^\ast(\rho)$ is less than or equal to that of $\rho$. Thus  $\psi_2^\ast(\rho)$ is an object in $(\Delta_{/X})^{\leq n}$, and so $\psi_1:\sigma\to \psi_2^\ast(\rho)$ defines an object in $(\Delta_{/X})^{\leq n}\times_{\Delta_{/X}}(\Delta_{/X})_{\sigma/}$
	
	Then we see that the commutative diagram 
	\[
	\begin{tikzcd}
		& \psi_2^\ast(\rho)\arrow[dd,"\psi_2"] \\
		\sigma\arrow[ur,"\psi_1"]\arrow[dr,"\psi"']	& \\
		& \rho 
	\end{tikzcd}
	\]
	defines a morphism in $(\Delta_{/X})^{\leq n}\times_{\Delta_{/X}}(\Delta_{/X})_{\sigma/}$. To show that $(\Delta_{/X})^{\leq n}\times_{\Delta_{/X}}(\Delta_{/X})_{\sigma/}$ is connected, it thus suffices to show that there is a morphism relating $\phi:\sigma\to \tau$ to $\psi_1:\sigma\to \psi_2^\ast(\rho)$. 
	
	If $\psi_2^\ast(\rho)$ is non-degenerate, then the uniqueness of the representation $(\phi,\tau)$ for $\sigma$ implies that $\psi_2^\ast(\rho)=\tau$ and $\psi_1=\phi$, so we are done. Otherwise, let $(\alpha,\gamma$ be the unique representation of $\psi_2^\ast(\rho)$ guaranteed by the Eilenberg-Zilber Lemma \ref{lem:E-Z}. Then note that $\alpha\circ \psi_1$ is surjective, $\gamma$ is non-degenerate, and 
	\[
	(\alpha\circ \psi_1)^\ast(\gamma)= \psi_1^\ast(\alpha^\ast(\gamma))=\psi_1^\ast(\psi_2^\ast(\rho))=\sigma. 
	\]
	Thus, the uniqueness guaranteed by the Eilenberg-Zilber Lemma \ref{lem:E-Z} implies that $\alpha\circ \psi_1=\phi$ and $\gamma=\tau$. We thus have a commutative diagram 
	\[
	\begin{tikzcd}
		& \psi_2^\ast(\rho)\arrow[dd,"\alpha"] \\
		\sigma\arrow[ur,"\psi_1"]\arrow[dr,"\phi"']	& \\
		& \gamma  
	\end{tikzcd}
	\]
	which defines a morphism in $(\Delta_{/X})^{\leq n}\times_{\Delta_{/X}}(\Delta_{/X})_{\sigma/}$. This completes the proof.
\end{proof}

\section{Relative coskeletalness}

We first recall classic coskeletalness. 

\begin{defn}
	For $k\geq 0$, the category $\Delta_{\leq k}$ is the full subcategory on the objects $[n]$  with $n\leq k$. The \emph{$k$\textsuperscript{th} truncation functor} is the functor 
	\[
	\begin{tikzcd}
		\on{tr}_k: &[-3em] \catsSet \arrow[r] & \on{Fun}(\Delta_{\leq k}^\op, \catSet)
	\end{tikzcd}
	\] 
	given by restriction along the inclusion $\Delta_{\leq k}\to \Delta$. The right adjoint to $\on{tr}_k$, given explicitly by right Kan extension, is called the \emph{$k$\textsuperscript{th} coskeleton functor}, and denoted by 
	\[
	\begin{tikzcd}
		\on{cosk}_k:&[-3em] \on{Fun}(\Delta_{\leq k}^\op, \catSet) \arrow[r] & \catsSet.
	\end{tikzcd}
	\]
	Since $\on{cosk}_k$ is the right Kan extension along a full subcategory inclusion, the adjunction counit $\on{tr}_k(\on{cosk}_k(X))\to X$ is an isomorphism. We will typically use the notation 
	\[
	\catsSet_{\leq k} =\on{Fun}(\Delta_{\leq k}^\op, \catSet).
	\]
\end{defn} 

\begin{defn}
	A simplicial set $X$ is called \emph{$k$-coskeletal} if, for every $n>k$ every extension problem 
		\[
	\begin{tikzcd}
		\partial \Delta^n \arrow[r]\arrow[d] & X  \\
		\Delta^n & 
	\end{tikzcd}
	\]
	has a unique solution. 
\end{defn}

\begin{pro}
	The adjunction unit $X \to \on{cosk}_k(\on{tr}_k(X))$ is an isomorphism if and only if $X$ is $k$-coskeletal. Thus, $\on{cosk}_k$ is fully faithful, and its essential image consists of the $k$-coskeletal simplicial sets. 
\end{pro}

\begin{proof}
	This is a classical fact. A proof may be found in, e.g., \cite[Lemma 10.11]{cosk}. 
\end{proof}

We can then turn to the relative case.

\begin{defn}
	Let $\pi:X\to Y$ be a map of simplicial sets. We call $X$ \emph{$\pi$-relative $k$-coskeletal} (or we call $\pi$ \emph{relatively $k$-coskeletal}) if, for any $n>k$, every lifting problem 
	\[
	\begin{tikzcd}
		\partial \Delta^n \arrow[r]\arrow[d] & X\arrow[d,"\pi"]  \\
		\Delta^n \arrow[r] & Y
	\end{tikzcd}
	\]
	has a unique solution. 
\end{defn}

We now concentrate on principal bundles with structure group $N(H)$ for $H$ an Abelian group. 

\begin{lem}\label{lem:N(H)-bundle_rel_2-cosk}
	Let $X$ be a simplicial set, and let $\eta$ be a $N(H)$-valued twisting function on $X$. Then $\pi_\eta: E_\eta\to X$ is relatively 2-coskeletal. 
\end{lem}

\begin{proof}
	Solving a lifting problem 
	\[
	\begin{tikzcd}
		\partial \Delta^n \arrow[r]\arrow[d] & E_\eta\arrow[d,"\pi_\eta"]  \\
		\Delta^n \arrow[r] & X
	\end{tikzcd}
	\]
	is equivalent to solving the lifting problem 
	\[
	\begin{tikzcd}
		\partial \Delta^n \arrow[r]\arrow[d] & E_\eta \times_{X} \Delta^n\arrow[d] \\
		\Delta^n \arrow[r,"="'] & \Delta^n 
	\end{tikzcd}
	\]
	Since every principal bundle over $\Delta^n$ is trivial, this is equivalent to solving the lifting problem
	\[
	\begin{tikzcd}
		\partial \Delta^n \arrow[r]\arrow[d] & N(H)\times \Delta^n\arrow[d] \\
		\Delta^n \arrow[r,"="'] & \Delta^n 
	\end{tikzcd}
	\]
	and thus, equivalent to showing that any diagram
	\[
	\begin{tikzcd}
		\partial \Delta^n \arrow[r]\arrow[d] & N(H) \\
		\Delta^n &
	\end{tikzcd}
	\]
	has a unique extension. However, since nerves of categories are always 2-coskeletal, this is immediate. 
\end{proof}

\begin{lem}\label{lem:rel_cosk_unique_ext}
	Let $\pi_E:E\to X$ and $\pi_F:F\to X$ be maps, and suppose that $F$ is $\pi_F$-relatively $k$-coskeletal. Suppose given a simplicial map $f:\on{tr}_k(E)\to \on{tr}_k(F)$ over $\on{tr}_k(X)$. Then $f$ uniquely extends to a map of simplicial sets over $X$. 
\end{lem}

\begin{proof}
	It will suffice to show that $f$ extends uniquely to the $(k+1)$-truncation. For each non-degenerate $(k+1)$ simplex $\sigma:\Delta^{k+1}\to E$, we get an associated lifting problem
	\[
	\begin{tikzcd}
		\partial \Delta^n \arrow[r,"f\circ \partial \sigma"]\arrow[d] & F\arrow[d,"\pi_F"] \\
		\Delta^n \arrow[r,"\pi_E\circ \sigma"'] & X
	\end{tikzcd}
	\]
	any $\widetilde{f}$ extending $f$ to the $(k+1)$ truncation must send $\sigma$ to a solution to this lifting problem. Since $F$ is relatively $k$-coskeletal, such a solution is unique, and so there is at most one such $\widetilde{f}$. On the other hand, defining $\widetilde{f}(\sigma)$ to be the unique solution to this lifting problem yields the desired extension. 
\end{proof}

\begin{cor}\label{cor:relcosk_iso_tr_iso}
	Let 
	\[
	\begin{tikzcd}
		E \arrow[rr,"f"]\arrow[dr,"{\pi_E}"'] && F\arrow[dl,"\pi_F"] \\
		& X &
	\end{tikzcd}
	\]
	be a map between relatively $k$-coskeletal simplicial sets over $X$. If $f$ is an isomorphism on $k$-truncations, then $f$ is an isomorphism. 
\end{cor}

\begin{proof}
	Since $f$ is an isomorphism on $k$-truncations, the previous lemma implies its inverse on $k$-truncations has a unique extension to a map $g:F\to E$ over $X$. The uniqueness of extensions immediately shows that $g\circ f=\on{id}_E$ and $f\circ g=\on{id}_F$. 
\end{proof}

To complete our discussion of relative coskeletalness, we can simply duplicate the characterization of coskeletal simplicial sets in terms of an adjunction. Since slices of presheaf categories are equivalent to presheaf categories on over-categories, we can identify 
\[
\catsSet\simeq \on{Fun}(\Delta_{/X}^\op,\catSet) \quad \text{and} \quad \catsSet_{\leq k}\simeq \on{Fun}((\Delta_{\leq k})_{/\on{tr}_k(X)}^\op,\catSet).
\]

\begin{defn}
	The \emph{$k$\textsuperscript{th} relative truncation functor over $X$} is the functor 
	\[
	\begin{tikzcd}
		\on{tr}_k^X:&[-3em] \catsSet\simeq \on{Fun}(\Delta_{/X}^\op,\catSet) \arrow[r]&  \on{Fun}((\Delta_{\leq k})_{/\on{tr}_k(X)}^\op,\catSet)\simeq \catsSet_{\leq k}
	\end{tikzcd}
	\]
	given by composing with the inclusion $(\Delta_{\leq k})_{/\on{tr}_k(X)}\to \Delta_{/X}$. The \emph{$k$\textsuperscript{th} relative coskeleton functor over $X$} is the right adjoint $\on{cosk}_k^X$ to $\on{tr}_k^X$, given by right Kan extension. Note that the adjunction counit is an isomorphism. 
\end{defn}

\begin{pro}
	For $X\in \catsSet$, $\pi:E\to X$ is relative $k$-coskeletal if and only if the adjunction unit
	\[
	\begin{tikzcd}
		E \arrow[rr]\arrow[dr,"\pi"'] & & \on{cosk}_k^X(\on{tr}_k^X(E))\arrow[dl,"\on{tr}(\pi)"]\\
		 & X & 
	\end{tikzcd}
	\]
	is an isomorphism. 
\end{pro}

\begin{proof}
	Since the $k$\textsuperscript{th} truncation of $\partial \Delta^n \to \Delta^n$ is an isomorphism whenever $n>k$, it is immediate that, for any $k$-truncated simplicial set $Y$ over $\on{tr}_k(X)$, the map $\on{cosk}_k^X(Y)\to X$ is relative $k$-coskeletal. Thus, if the adjunction unit is an isomorphism, $E$ is relative $k$-coskeletal. 
	
	On the other hand, since the counit of the adjunction is an isomorphism, the induced map 
	\[
	\begin{tikzcd}
		\on{tr}_k^X(E) \arrow[r] & \on{tr}_k^X(\on{cosk}_k^X(\on{tr}_k^X(E)))
	\end{tikzcd}
	\]
	is an isomorphism. Thus, by Cor \ref{cor:relcosk_iso_tr_iso},
	 if $E$ is relative $k$-coskeletal over $X$, the adjunction unit is an isomorphism. 
\end{proof}

\begin{cor}
	Let $X$ be a simplicial set, and $\underline{\pi}:\underline{E}\to \on{tr}_k(X)$ be a $k$-truncated simplicial set over $E$. There is a relative $k$-coskeletal simplicial set $\pi:E\to X$ over $X$ whose $k$-truncation is $\underline{\pi}:\underline{E}\to \on{tr}_k(X)$. Moreover, $E$ is unique up to unique isomorphism over $X$. 
\end{cor}

\begin{proof}
	The uniqueness up to unique isomorphism is immediate from corollary \ref{cor:relcosk_iso_tr_iso}, and existence follows by applying $\on{cosk}_k^X$.  
\end{proof}

\section{Proofs by simplicial computations}\label{sec:simp_comp}

\begin{proof}[Proof (of Lemma \ref{lem:classifying_map_of_twisting_function})]
	Equivariance is immediate from the definition. We check the simplicial compatibilities explicitly. Firstly, for $i<n$, we compute 
	\[
	\begin{aligned}
		d_i\theta^\eta_n(x,g) &= d_i \left(g,\eta_n(x),\eta_{n-1}(d_0x),\ldots,\eta_1(d_0^{n-1}(x))\right)\\
		&= \left(d_ig_n,d_{i-1}\eta_n(x),\ldots,d_0(\eta_{n-i+1}(d_0^{i-1}x))\eta_{n-i}(d_0^{i}(x)),\ldots,\eta_1(d_0^{n-1}(x))\right)\\
		&=\left(d_ig_n,\eta_{n-1}(d_ix),\ldots,\eta_{n-i}(d_1d_0^{i-1}(x))\eta_{n-i}(d_0^{i}(x))^{-1}\eta_{n-i}(d_0^{i}(x)),\ldots,\eta_1(d_0^{n-1}(x))\right)\\
		&=\left(d_ig_n,\eta_{n-1}(d_ix),\ldots,\eta_{n-i}(d_0^{i}(x)),\ldots,\eta_1(d_0^{n-1}(x))\right)\\
		&=\left(d_ig_n,\eta_{n-1}(d_ix),\ldots,\eta_{n-i}(d_0^{i}(x)),\ldots,\eta_1(d_0^{n-2}(d_ix))\right)= \theta^\eta_{n-1}(d_ix,d_ig)
	\end{aligned}
	\] 
	If $i=n$, we compute
	\[
	\begin{aligned}
		d_n\theta^\eta_n(x,g) &= d_n \left(g,\eta_n(x),\eta_{n-1}(d_0x),\ldots,\eta_1(d_0^{n-1}(x))\right)\\
		&=\left(d_ng_n,d_{n-1}\eta_{n}(x),\ldots, d_1\eta_2(d_0^{n-2}x)\right)\\
		&= \left(d_ng_n,\eta_{n-1}(d_nx),\ldots, \eta_1(d_2d_0^{n-2}x)\right)\\
		&= \left(d_ng_n,\eta_{n-1}(d_nx),\ldots, \eta_1(d_0^{n-2}d_nx)\right)\\
		&= \theta_{n-1}^{\eta}(d_nx,d_ng)
	\end{aligned}
	\]
	Finally, we compute 
	\[
	\begin{aligned}
		s_i \theta_n^{\eta}(x,g) & = s_i\left(g,\eta_n(x),\eta_{n-1}(d_0x),\ldots,\eta_1(d_0^{n-1}(x))\right)\\
		&= \left(s_ig,s_{i-1}\eta_n(x),s_{i-2}\eta_{n-1}(d_0x),\ldots,s_0\eta_{n-i+1}(d_0^{i-1}(x)),1,\eta_{n-i}(d_0^i(x)),\ldots,\eta_1(d_0^{n-1}(x))\right)\\
		&= \left(s_ig,\eta_{n+1}(s_ix),\eta_{n}(s_{i-1}d_0x),\ldots,\eta_{n-i+2}(s_1d_0^{i-1}(x)),1,\eta_{n-i}(d_0^i(x)),\ldots,\eta_1(d_0^{n-1}(x))\right)\\
		&= \left(s_ig,\eta_{n+1}(s_ix),\eta_{n}(d_0s_ix),\ldots,\eta_{n-i+2}(d_0^{i-1}s_i(x)),1,\eta_{n-i}(d_0^i(x)),\ldots,\eta_1(d_0^{n-1}(x))\right)\\
		&= \left(s_ig,\eta_{n+1}(s_ix),\eta_{n}(d_0s_ix),\ldots,\eta_{n-i+2}(d_0^{i-1}s_i(x)),1,\eta_{n-i}(d_0^{i+1}(s_ix)),\ldots,\eta_1(d_0^{n}(s_ix))\right)\\
		&= \left(s_ig,\eta_{n+1}(s_ix),\eta_{n}(d_0s_ix),\ldots,\eta_{n-i+2}(d_0^{i-1}s_i(x)),\eta_{n-1+1}(s_0(d_0^i(x))),\eta_{n-i}(d_0^{i+1}(s_ix)),\ldots,\eta_1(d_0^{n}(s_ix))\right)\\
		&= \left(s_ig,\eta_{n+1}(s_ix),\eta_{n}(d_0s_ix),\ldots,\eta_{n-i+2}(d_0^{i-1}s_i(x)),\eta_{n-1+1}((d_0^i(s_ix))),\eta_{n-i}(d_0^{i+1}(s_ix)),\ldots,\eta_1(d_0^{n}(s_ix))\right)\\
		&= \theta^\eta_{n+1}(s_ix,s_ig).
	\end{aligned}
	\]
	so the map $\theta^\eta$ is simplicial. Moreover, since $\theta$ is levelwise an equivariant map of free $G$-sets, the diagram is pullback, completing the proof. 
\end{proof}

\begin{proof}[Proof (of Lemma \ref{lem:twist_fxn_norm_cocycle})]
		We need only check the defining relations. The inductive definition is made precisely so that 
	\[
	d_0(\eta(x))=\eta(d_1(x))-\eta(d_0(x)).
	\]
	For the other face map relations, note that they trivially hold for $\eta_2$. Then inductively suppose they hold for $k\leq n-1$. Then for $i>1$ we have  
	\[
	\begin{aligned}
		d_i(\eta_n(x)) & =\left(\gamma(d_3\cdots d_n(x)),d_{i-1}(\eta_{n-1}(d_1x)-\eta_{n-1}(d_0x))\right)\\
		& = \left(\gamma(d_3\cdots d_n(x)),(\eta_{n-2}(d_id_1x)-\eta_{n-2}(d_id_0x))\right)\\
		&= \left(\gamma(d_3\cdots d_n(x)),(\eta_{n-2}(d_1d_{i+1}x)-\eta_{n-2}(d_0d_{i+1}x))\right)\\
		&= \eta_{n-1}(d_{i+1}(x)) 
	\end{aligned}
	\]
	as desired. Finally, note that 
	\[
	\begin{aligned}
		d_1(\eta_n(x)) &= \left(\gamma(d_3\cdots d_n(x))+\gamma(d_3\cdots d_{n-1}(d_1x))-\gamma(d_3\cdots d_{n-1}(d_0x)), d_0(\eta_{n-1}(d_1x)-\eta_{n-1}(d_0x))\right)\\
		&= \left(\gamma(d_2d_4\cdots d_{n}x), d_0(\eta_{n-1}(d_1x)-\eta_{n-1}(d_0x))\right)\\
		&= \left(\gamma(d_3\cdots d_{n-1}(d_2x)),d_0(\eta_{n-1}(d_1x)-\eta_{n-1}(d_0x))\right)\\
		&= \left(\gamma(d_3\cdots d_{n-1}(d_2x)),\eta_{n-2}(d_1d_1x)-\eta_{n-2}(d_1d_0x)-\eta_{n-2}(d_0d_1x)+\eta_{n-2}(d_0d_0x)\right)\\
		&=  \left(\gamma(d_3\cdots d_{n-1}(d_2x)),\eta_{n-2}(d_1d_2x)-\eta_{n-2}(d_0d_2x)-\eta_{n-2}(d_0d_0x)+\eta_{n-2}(d_0d_0x)\right)\\
		&=\left(\gamma(d_3\cdots d_{n-1}(d_2x)),\eta_{n-2}(d_1d_2x)-\eta_{n-2}(d_0d_2x)\right)\\
		&= \eta_{n-1}(d_2(x)).
	\end{aligned}
	\]
	To see the relations with the $s_i$, we similarly argue inductively. First, in the case $i=0$
	\[
	\begin{aligned}
		\eta_n(s_0(x)) &= \left(\gamma(d_3\cdots d_n(s_0(x))),\eta_{n-1}(d_1s_0(x))-\eta_{n-1}(d_0s_0x) \right)\\
		& = \left(\gamma(d_3\cdots d_n(s_0(x))),0 \right)\\
		&= \left(\gamma(s_0d_2\cdots d_{n-1}(x)),0 \right)\\
		&= \left(0,0 \right)
	\end{aligned}
	\]
	where the final equality follows because $\gamma$ is normalized. For $i=1$, we compute 
	\[
	\begin{aligned}
		\eta_n(s_1(x)) &= \left(\gamma(d_3\cdots d_n(s_1(x))),\eta_{n-1}(d_1s_1(x))-\eta_{n-1}(d_0s_1x) \right)\\
		&= \left(\gamma(s_1d_2\cdots d_{n-1}((x))),\eta_{n-1}(x)-\eta_{n-1}(s_0d_0x) \right)\\
		&= \left(0,\eta_{n-1}(x)-0 \right)\\
		&=s_0(\eta_{n-1}(x)).
	\end{aligned}
	\]
	Finally, when $i>1$, we compute 
	\[
	\begin{aligned}
		\eta_n(s_i(x)) &= \left(\gamma(d_3\cdots d_n(s_i(x))),\eta_{n-1}(d_1s_i(x))-\eta_{n-1}(d_0s_ix) \right)\\
		& =\left(\gamma(d_3\cdots d_n(s_i(x))),\eta_{n-1}(s_{i-1}d_1(x))-\eta_{n-1}(s_{i-1}d_0x) \right)\\
		&=\left(\gamma(d_3\cdots d_{n-1}(x)),\eta_{n-1}(s_{i-1}d_1(x))-\eta_{n-1}(s_{i-1}d_0x) \right)\\
		&= \left(\gamma(d_3\cdots d_{n-1}(x)),s_{i-2}\eta_{n-1}(d_1(x))-s_{i-2}\eta_{n-1}(d_0x) \right)\\
		&= s_{i-1} \left(\gamma(d_3\cdots d_{n-1}(x)),\eta_{n-1}(d_1(x))-\eta_{n-1}(d_0x) \right)\\
	\end{aligned}
	\]
	completing the proof.
\end{proof}

\end{document}